\documentclass[11pt]{article}

	\usepackage{enumerate}
	\usepackage{natbib} %comment out if you do not have the package
\usepackage{geometry}
\usepackage{fancyhdr}
\usepackage{afterpage}
\usepackage{graphicx}
\usepackage{tocloft}
\usepackage{float}
\usepackage{subcaption}
\usepackage{amsmath,amssymb,amsthm}
\usepackage{multicol}
\usepackage{mathtools}
\usepackage{dcolumn,array}
\usepackage{color}
\usepackage{multirow}
\usepackage{tikz}
\usetikzlibrary{arrows}
\usepackage{adjustbox}
\usepackage{makecell}
\usepackage{amsfonts}
\usepackage{placeins}
\usepackage{setspace}
\usepackage[font=small,labelfont={bf}]{caption}
\usepackage[ruled,vlined]{algorithm2e}
\usepackage{longtable}  % For tables that span multiple pages

\newcommand{\bla}{\color{black} }
\newcommand{\cT}{\mathcal{T}}
\newcommand{\cV}{\mathcal{V}}
\newcommand{\cW}{\mathcal{W}}
\newcommand{\Gigantic}[1]{\fontsize{#1}{#1-3}\selectfont}
\newtheorem{Proposition}{Proposition}
\newtheorem{Definition}{Definition}
\newtheorem{Theorem}{Theorem}

\newtheorem{Example}{Example}

\newcolumntype{"}{@{\hskip\tabcolsep\vrule width 2pt\hskip\tabcolsep}}
\SetKwInOut{KwIn}{Input}
\SetKwInOut{KwOut}{Output}
\SetKw{KwRet}{Return}
\newcounter{cases}
\newcounter{subcases}[cases]

\newtheorem{manualtheoreminner}{Theorem}
\newenvironment{manualtheorem}[1]{%
  \renewcommand\themanualtheoreminner{#1}%
  \manualtheoreminner
}{\endmanualtheoreminner}
\newtheorem{manualPropositioninner}{Proposition}
\newenvironment{manualProposition}[1]{%
  \renewcommand\themanualPropositioninner{#1}%
  \manualPropositioninner
}{\endmanualPropositioninner}

\newcommand{\lV}{\vspace{-1mm}}
\newcommand{\llV}{\vspace{-2mm}}
\newcommand{\lllV}{\vspace{-5mm}}

\newcommand{\mmV}{\vspace{2mm}}
\newcommand{\mmmV}{\vspace{5mm}}

\begin{document}
		
		%	%%%%%%%%%%%%%%%%%%%%%%%%%%%%%%%%%%%%%%%%%%%%%%%%%%%%%%%%%%%%%%%%%%%%%%%%%%%%%%
		\def\spacingset#1{\renewcommand{\baselinestretch}%
			{#1}\small\normalsize} \spacingset{1}
%		%%%%%%%%%%%%%%%%%%%%%%%%%%%%%%%%%%%%%%%%%%%%%%%%%%%%%%%%%%%%%%%%%%%%%%%%%%%%%%

			\title{\bf \emph{The Bidirectionality-Preserving Feedback Arc Set Problem with an Application to Faculty Hiring Network Analysis}}
			\author{Sina Akbari $^a$, Adolfo R. Escobedo $^b$, and Jorge A. Sefair $^c$ \\
			$^a$ School of Computing and Augmented Intelligence,\\ Arizona State University, Tempe, AZ, U.S.A\\
             $^b$ Edward P. Fitts Department of Industrial and Systems Engineering,\\ North Carolina State University, Raleigh, NC, U.S.A.\\
             $^c$ Department of Industrial and Systems Engineering,\\ University of Florida, Gainesville, FL, U.S.A}
			\date{}
			\maketitle

%		\if1\blind
%		{
%
%            \title{\bf \emph{The Bidirectionality-Preserving Feedback Arc Set Problem with an Application to Faculty Hiring Network Analysis}}
%			\author{Author information is purposely removed for double-blind review}
%			
%\bigskip
%			\bigskip
%			\bigskip
%			\begin{center}
%				{\LARGE\bf \emph{The Bidirectionality-Preserving Feedback Arc Set Problem with an Application to Faculty Hiring Network Analysis}}
%			\end{center}
%			\medskip
%		} \fi
%		%\bigskip
%		\lllV
	\begin{abstract}
This work introduces a novel combinatorial optimization problem to uncover strict and non-strict hierarchical structures underlying digraphs with bidirected arcs. We show that the problem includes the seminal feedback arc set problem as a special case, and that its solution extends a tournament solution known as Slater's rule. We introduce and compare two binary programming formulations and also introduce a polynomial-time search algorithm for obtaining an ordered hierarchy from their solution. The methodology is applied to the faculty hiring network problem using a dataset previously collected by Del Castillo et al. (2020) to  characterize the prestige/rank of Industrial/Systems/Operations Research (IEOR) departments using faculty hiring transactions.
	\end{abstract}
			
	\noindent%
	{\it Keywords:} \emph{feedback arc set; hiring networks; university rankings; tournament graphs}.

	%\newpage
	\setstretch{1.4} % DON'T change the spacing!

\section{Introduction}\label{sec: Introduction}
The \textit{feedback arc set problem (FASP)} is one of the most prominent problems in computer science and graph theory. Given a digraph, a \textit{feedback arc set} is a subset of arcs whose removal makes the resulting graph acyclic. The \textit{minimum feedback arc set (MFAS)} is the feedback arc set of minimum cardinality and cumulative weight for unweighted and weighted digraphs, respectively.  
FASP has wide-ranging applications across various fields including detecting deadlock in the design of operating systems \citep{bic1988logical}, finding favorable computation sequences in process flowsheet calculations in chemical engineering \citep{baharev2021exact}, and determining winners of a tournament in social choice theory \citep{slater1961inconsistencies}. 

FASP is interlinked with the \textit{linear ordering problem (LOP)} \cite{grotschel1985acyclic}. The two problems can be interpreted as seeking to uncover an underlying ordering of a set of given entities that is closest in a precise mathematical sense to a set of given orderings over the entities. In both cases, the inputs can be represented via a digraph, where the nodes represent the entities to be ordered, and the arcs summarize the given ordinal relationships. 

LOP seeks to order all of the given entities strictly; its inputs are also required to be linear orderings, which effectively implies that every pair of entities were compared the same number of times. In FASP, on the other hand, the inputs can come from more granular pairwise preferences, and it is not required for every pair to be compared directly or indirectly in the inputs. This makes FASP applicable to a wider array of situations where LOP is inadequate, including those when the given orderings are sparse and incomplete. It also implies that, unlike LOP, the FASP solution (i.e., the input digraph after the feedback arc set is removed), may not readily induce a linear ordering of the entities. However, topological sorting algorithms such as depth-first search (DFS) \citep{tarjan1972depth} can be efficiently deployed onto the resulting directed acyclic graph to find them.

A major drawback of FASP (and LOP) is that each subset of entities that is ordered must be done so strictly, and, as such, it is not possible to characterize another key hierarchical relationship of interest: equivalence classes. In other words, FASP hinders the ability to uncover subgroups of entities with equal standing among one another. Although the problem admits ties in the input graph---encoded via bidirectional arcs between pairs of entities---these non-strict ordering relationships are always broken through the solution process. This is a significant limitation that could cause unjustified ordinal relationships to be reported, for example, an entity being ordered ahead of others in its equivalence class. Such concerns can be circumvented in certain cases via the weak ordering problem (WOP), which seeks the total ordering of the entities that is closest to the given orderings, allowing for tied relationships in both the inputs and outputs \citep{esc23der}. However, analogous to the LOP, WOP requires every pair entities to be compared in the given orderings---i.e., the input digraph must be complete---which may be impractical to satisfy in various real-world settings.

As a case in point, consider graph-based methodologies for ranking university departments of specific academic disciplines, which are advocated as rigorous and transparent alternatives to ad hoc approaches employed by private institutions such as U.S. News and World Report (U.S.NEWS). These alternative approaches seek to analyze a \emph{faculty hiring network (FHN)}, which can be represented as a weighted digraph where each academic department is represented via a node, and the weight of each directed arc $i$ to $j$ is equal to the number of Ph.D. graduates that department $i$ has placed at as a tenure-track faculty member at department $j$. The FHN approach is built on the notion that the rank of an academic department can be derived from the placements of its graduates into tenure-track faculty positions at other departments, based on the authority and prestige that such transactions convey \citep{way2016gender}. 

Faculty hiring transactions among universities induce relatively sparse digraphs. In fact, recent studies expose that most newly hired faculty members are graduates from just a few institutions. For instance, in the U.S., only 19–28\% of universities educate 80\% of faculty members across various disciplines \citep{wapman2022quantifying}. Hiring transactions between departments can occur in both directions, meaning that the input digraph should admit bidirected arcs. Based on these characteristics, optimization approaches that rely on the solution to FASP have been devised to study these networks \cite{del2020exponential}. However, the aforementioned limitations of the problem preclude the characterization of hierarchies where certain departments have relatively indistinguishable differences in prestige. Such cases are relevant, for example, when there is a relative balance in the mutual hiring transactions between two or more departments. To tackle this and other settings where uncovering non-strict hierarchical relationships is of interest, we introduce a generalized version of FASP and a polynomial search algorithm for obtaining a full non-strict ordering from its solution. We make various additional fundamental contributions including establishing formal connections with Slater's Rule, a well-known tournament solution in social choice \citep{brandt2016handbook}, and deriving and comparing two mathematical formulations of the problem of interest. Finally, we illustrate the proposed methodology by studying the U.S. Industrial/Systems/Operations Research (IEOR) faculty hiring network \citep{del2020exponential}. 

The rest of this paper is organized as follows. Section \ref{Sec: Definitions} provides definitions and other preliminaries. Section \ref{Sec: Generalized Feedback Arc Set Problem} introduces the \emph{Bidirectionality-Preserving Feedback Arc Set Problem (BP-FASP)}, explores connections to other graph-theoretic concepts, and derives a topological search algorithm for extracting a total ordering from the problem's solution. In addition, it leverages BP-FASP to develop an extension of Slater's Rule to weak tournaments. Section \ref{Sec: Mathematical Formulations} introduces two mathematical formulations and compares their strengths via polyhedral analysis. Section \ref{Sec: Case Study} applies the proposed methodology to study the IEOR FHN, comparing the results with the original case study in \citet{del2020exponential}. Finally, Section \ref{Sec: Conclusion} concludes the paper.\llV

\section{Definitions}\label{Sec: Definitions}
Let $\mathcal{G} = (\mathcal{N}, \mathcal{A})$ be a weighted simple digraph, where $\mathcal{N}$ is a set of nodes and $\mathcal{A} \subseteq \mathcal{N} \times \mathcal{N}$ is a set of arcs; furthermore, let $\mathcal{W} \in [0, \infty)^{|\mathcal{A}|}$ be the set of arc weights. As a convention, let $(i, j)$ denote the arc directed from node $i$ to node $j$, and let $w_{ij} \in \mathcal{W}$ denote its weight. It is said that $\mathcal{G}$ is a simple digraph if it has no multiple edges or self-loops, and that $\mathcal{G}$ is a complete digraph if there is a bidirectional arc between every pair of nodes. Fig. \ref{Complete digraphs of size 1, 2, 3, and 4} depicts the complete digraphs of size 1, 2, 3, and 4.\llV
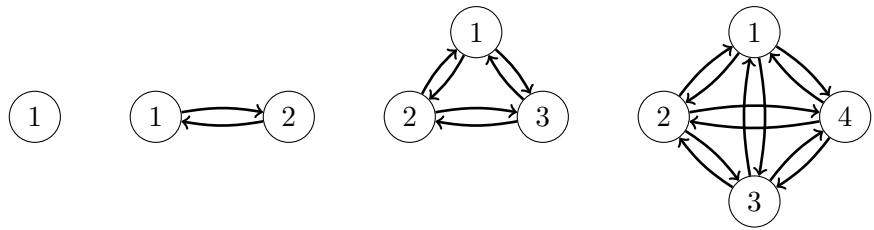
\begin{figure}[th]
\centering
\begin{tikzpicture}[scale=0.8]
\tikzstyle{every node}=[draw, shape=circle];
\path (0, 0cm) node (1) {$1$};
\tikzstyle{every node}=[draw, shape=circle];
\path (2, 0cm) node (2) {$1$};
\path (4.2, 0cm) node (1) {$2$};
\draw[line width=1pt,->]
(1) to [bend left=10] (2);
\draw[line width=1pt,->]
(2) to [bend left=10] (1);
\tikzstyle{every node}=[draw, shape=circle];
\path (6.2, 0cm) node (2) {$2$};
\path (8.4, 0cm) node (3) {$3$};
\path (7.3, 1.4cm) node (1) {$1$};
\draw[line width=1pt,->]
(1) to [bend left=10] (2);
\draw[line width=1pt,->]
(2) to [bend left=10] (1);
\draw[line width=1pt,->]
(2) to [bend left=10] (3);
\draw[line width=1pt,->]
(3) to [bend left=10] (2);
\draw[line width=1pt,->]
(3) to [bend left=10] (1);
\draw[line width=1pt,->]
(1) to [bend left=10] (3);
\tikzstyle{every node}=[draw, shape=circle];
\path (10.4, 0cm) node (2) {$2$};
\path (13.4, 0cm) node (4) {$4$};
\path (11.9, 1.4cm) node (1) {$1$};
\path (11.9, -1.4cm) node (3) {$3$};
\draw[line width=1pt,->]
(1) to [bend left=10] (2);
\draw[line width=1pt,->]
(2) to [bend left=10] (1);
\draw[line width=1pt,->]
(1) to [bend left=10] (4);
\draw[line width=1pt,->]
(2) to [bend left=10] (3);
\draw[line width=1pt,->]
(3) to [bend left=10] (2);
\draw[line width=1pt,->]
(3) to [bend left=10] (4);
\draw[line width=1pt,->]
(4) to [bend left=10] (3);
\draw[line width=1pt,->]
(4) to [bend left=10] (1);
\draw[line width=1pt,->]
(3) to [bend left=10] (1);
\draw[line width=1pt,->]
(1) to [bend left=10] (3);
\draw[line width=1pt,->]
(4) to [bend left=10] (2);
\draw[line width=1pt,->]
(2) to [bend left=10] (4);
\end{tikzpicture}
\caption{Complete digraphs of sizes 1, 2, 3, and 4}
\label{Complete digraphs of size 1, 2, 3, and 4}
\end{figure}

Let $\overline{\mathcal{G}} = (\overline{\mathcal{N}}, \overline{\mathcal{A}})$ denote a subgraph of $\mathcal{G}$ with $\overline{\mathcal{N}} \subseteq \mathcal{N}$ and $\overline{\mathcal{A}} \subseteq \mathcal{A}$. The subgraph is said to be \textit{induced} by $\overline{\mathcal{N}}$ if $\overline{\mathcal{A}}$ includes all of the edges of $\mathcal{A}$ with both endpoints in $\overline{\mathcal{N}}$. A (directed) \textit{path} is a sequence of arcs $\left((i_1, i_2), (i_2, i_3), \dots, (i_{\ell-1}, i_\ell)\right)$ that connects a sequence of nodes $i_1, i_2, i_3, \dots, i_{\ell-1}, i_\ell$. Let $i_1 \leadsto i_\ell$ denote a \textit{simple path} that starts from $i_1$ and ends at $i_\ell$. A path $i_1 \leadsto i_\ell$ forms a \textit{simple cycle} if all $i_v$ are distinct, for $2 \leq v \leq \ell$, and $i_1 = i_\ell$. $\mathcal{G}$ is a directed acyclic graph (DAG) if it contains no simple cycles. 

In social choice and decision theory, $i \succ j $ denotes that $i$ is strictly preferred to $j$; $i \approx j$ denotes that $i$ and $j$ are tied; and $i \succeq j$ denotes that $i$ is preferred over or tied with $j$, or equivalently that $j$ is not preferred over $i$ \cite{brandt2016handbook}. A binary relation $\boldsymbol{R} \subseteq \mathcal{K} \times \mathcal{K}$ is a \textit{linear ordering}---otherwise known as strict ordering---on some set $\mathcal{K}$ if it is transitive (if $i \succeq j$ and $j \succeq k$, then $i \succeq k$ for $i, j, k \in \mathcal{K}$), complete ($i \succeq j$ or $j \succeq i$ for $i \neq j \in \mathcal{K}$), and antisymmetric (if $i \succeq j$ and $j \succeq i$, then $i = j$ for all $i, j \in \mathcal{K}$). In more succinct terms, a linear ordering on $\mathcal{K}$ is a permutation of all the elements of $\mathcal{K}$. Moreover, $\boldsymbol{R}$ is a \textit{weak ordering}---otherwise known as non-strict ordering---on $\mathcal{K}$ if it is transitive and complete, but not necessarily antisymmetric. Linear orders are used to express strict preferences and weak orders are used to express non-strict preferences \cite{yoo2021new}; stated otherwise, weak orders are a generalization of linear orders where ties \emph{may} be present. Thus, every linear order is a weak order, but not every weak order is a linear order. For example, $\boldsymbol{\sigma^1} = \{b \succ d \succ a \succ c\}$ and $\boldsymbol{\sigma^2} = \{b \succ d \approx a \succ c\}$ are a linear ordering and weak ordering, respectively, on $\mathcal{K} := \{a, b, c, d\}$. Another useful representation of weak orderings is via \emph{ordered partitions}, where members of the same partition (i.e., equivalence class) are considered as being tied. For example, $\boldsymbol{\sigma^1}$ and $\boldsymbol{\sigma^2}$ also can be expressed as $ \{\{b\}, \{d\}, \{a\}, \{c\}\}$ and $\{\{b\}, \{d, a\}, \{c\}\}$, respectively, a representation alternatively known as \emph{bucket orders} (see \cite{fagin2006comparing} for more details). 

Since linear orders are transitive, LOP and FASP can be solved interchangeably. In a similar vein, this work leverages the transitivity of weak orders to solve a generalized version of FASP, whose  original definition is as follows.
\begin{Definition}\label{FASP}
(FASP) Let $\mathcal{G} = (\mathcal{N}, \mathcal{A})$ be a weighted simple digraph. A feedback arc set of $\mathcal{G}$, denoted as $\mathcal{A}' \subseteq \mathcal{A}$, is a subset of arcs whose removal makes the resulting subgraph $\mathcal{G'} = (\mathcal{N}, \mathcal{A}\backslash \mathcal{A}')$ acyclic (i.e., a DAG). The minimum feedback arc set of $\mathcal{G}$, abbreviated as MFAS$(\mathcal{G})$, has the minimum total arc weight among all feedback arc sets.
\end{Definition}
\begin{Example}
    Consider the unweighted digraph depicted in Fig. \ref{fig:sample_graph}; two cycles, $1 \leadsto 2 \leadsto 4  \leadsto 1$ and $2 \leadsto 3 \leadsto 2$,  are shown via green and blue arcs. Here, MFAS$(\mathcal{G}) = \{(4, 1), (3, 2)\}$, which is shown via red arcs in Fig. \ref{fig:mfas} (other solutions are omitted for simplicity). The resulting DAG after removing MFAS$(\mathcal{G})$ is shown in Fig. \ref{fig:resulting_graph}.
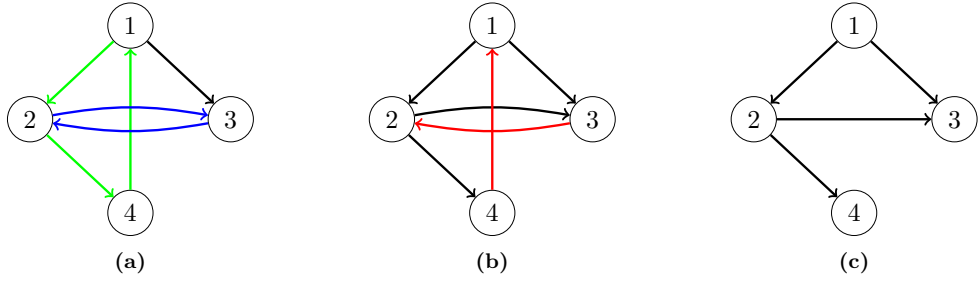
\begin{figure}[th]
    \centering
    \resizebox{0.9\textwidth}{!}{
    \begin{minipage}{\textwidth}
    \begin{tabular}{ccc}
        \begin{tabular}{@{}c@{}}
            \begin{subfigure}{0.33\textwidth}
                \centering
                \begin{tikzpicture}
                    \tikzstyle{every node}=[draw, shape=circle];
                    \path (10.4, 0cm) node (2) {$2$};
                    \path (13.4, 0cm) node (3) {$3$};
                    \path (11.9, 1.4cm) node (1) {$1$};
                    \path (11.9, -1.4cm) node (4) {$4$};
                    \draw[line width=1pt, green, ->] (1) to (2);
                    \draw[line width=1pt, green, ->] (4) to (1);
                    \draw[line width=1pt, blue, ->] (2) to [bend left=10] (3);
                    \draw[line width=1pt, blue, ->] (3) to [bend left=10] (2);
                    \draw[line width=1pt, ->] (1) to (3);
                    \draw[line width=1pt, green, ->] (2) to (4);
                \end{tikzpicture}
                \caption{}
                \label{fig:sample_graph}
            \end{subfigure}
        \end{tabular}
        &
        \begin{tabular}{@{}c@{}}
            \begin{subfigure}{0.33\textwidth}
                \centering
                \begin{tikzpicture}
                    \tikzstyle{every node}=[draw, shape=circle];
                    \path (10.4, 0cm) node (2) {$2$};
                    \path (13.4, 0cm) node (3) {$3$};
                    \path (11.9, 1.4cm) node (1) {$1$};
                    \path (11.9, -1.4cm) node (4) {$4$};
                    \draw[line width=1pt, ->] (1) to (2);
                    \draw[line width=1pt, ->] (2) to [bend left=10] (3);
                    \draw[line width=1pt, red, ->] (3) to [bend left=10] (2);
                    \draw[line width=1pt, ->] (1) to (3);
                    \draw[line width=1pt, ->] (2) to (4);
                    \draw[line width=1pt, red, ->] (4) to (1);
                \end{tikzpicture}
                \caption{}
                \label{fig:mfas}
            \end{subfigure}
        \end{tabular}
        & 
        \begin{tabular}{@{}c@{}}
            \begin{subfigure}{0.33\textwidth}
                \centering
                \begin{tikzpicture}
                    \tikzstyle{every node}=[draw, shape=circle];
                    \path (10.4, 0cm) node (2) {$2$};
                    \path (13.4, 0cm) node (3) {$3$};
                    \path (11.9, 1.4cm) node (1) {$1$};
                    \path (11.9, -1.4cm) node (4) {$4$};
                    \draw[line width=1pt, ->] (1) to (2);
                    \draw[line width=1pt, ->] (2) to (3);
                    \draw[line width=1pt, ->] (1) to (3);
                    \draw[line width=1pt, ->] (2) to (4);
                \end{tikzpicture}
                \caption{}
                \label{fig:resulting_graph}
            \end{subfigure}
        \end{tabular}
    \end{tabular}
    \caption{A sample graph and its MFAS, $\{(4, 1), (3, 2)\}$, shown in red arcs.}
    \label{fig:main_graph}
\end{minipage}
}
\end{figure}\llV

\end{Example}

To introduce the proposed generalization of FASP, we restate the definitions of unicycles and related concepts \cite{yoo2021new} and of an acyclic $K$-way partition \cite{herrmann2019multilevel}. A \textit{unicycle} is a simple path that starts and ends on the same node in one direction but not in the reverse direction. A \textit{unicyclic graph} contains at least one unicycle, and a \textit{unicycle-free graph} is devoid of such structures. Figs. \ref{Unicyclic graphs of size 3} and \ref{Unicycle-free graphs of size 3} depict the unicyclic and unicycle-free graphs of size three, respectively. Unicycle-free graphs are a generalization of DAGs: all DAGs are unicycle-free graphs, but not all unicycle-free graphs are DAGs. Specifically, DAGs admit only the trivially unicycle-free graphs of size three, which are shown in Appendix \ref{Appendix Trivial unicycle-free graphs of size 3}, and the unicycle-free graph depicted in Fig. \ref{fig: a} (and their isomorphisms); they exclude the other types of unicycle-free graph depicted in Figs. \ref{fig: b}, \ref{fig: c}, and \ref{fig: d}. 

\begin{Definition}\label{k-way}
(Acyclic $K$-Way Partition) Let $\mathcal{G} = (\mathcal{N}, \mathcal{A})$ be a simple digraph. A partition $\boldsymbol{\mathcal{N}} = \{\mathcal{N}_{1}, \mathcal{N}_{2},..., \mathcal{N}_{K}\}$ of $\mathcal{N}$ is called an acyclic $K$-way partition if at most one of the paths $i \leadsto j$ and $j' \leadsto i'$ exists, for $i, i' \in \mathcal{N}_u, \hspace{2pt} j, j' \in \mathcal{N}_v$, where $1 \leq u \not = v \leq K$.
\end{Definition}
\begin{Example}
Fig. \ref{a sample DAG and two partitions of its nodes} depicts two partitions of a DAG: $\boldsymbol{\mathcal{N}^1} = \{\mathcal{N}_{1}^{1} = \{1, 2\}, \mathcal{N}_{2}^{1} = \{3, 4\}\}$ and $\boldsymbol{\mathcal{N}^2} = \{\mathcal{N}_{1}^{2} = \{1, 3\}, \mathcal{N}_{2}^{2} = \{2, 4\}\}$. $\boldsymbol{\mathcal{N}^1}$ (center subfigure) is not an acyclic 2-way partition since the subsets are connected in both directions, namely via paths $1 \leadsto 3$ and $4 \leadsto 2$; on the other hand, $\boldsymbol{\mathcal{N}^2}$ (right subfigure) is an acyclic 2-way partition.
\end{Example}
\begin{figure}[H]
\centering
\begin{subfigure}{0.31\textwidth}
\begin{tikzpicture}
\tikzstyle{every node}=[draw, shape=circle];
\path (0, 0cm) node (j) {$j$};
\path (2.2, 0cm) node (k) {$k$};
\path (1.1, 1.6cm) node (i) {$i$};
\draw[line width=1pt,->] 
(i) to (j);
\draw[line width=1pt,->] 
(j) to (k);
\draw[line width=1pt,->]
(k) to (i);
\end{tikzpicture}\caption{}
\label{fig: f}
\end{subfigure}
\begin{subfigure}{0.31\textwidth}
\begin{tikzpicture}
\tikzstyle{every node}=[draw, shape=circle];
\path (4, 0cm) node (j) {$j$};
\path (6.2, 0cm) node (k) {$k$};
\path (5.1, 1.6cm) node (i) {$i$};
\draw[line width=1pt,->] 
(j) to [bend left=10] (k);
\draw[line width=1pt,->] 
(k) to [bend left=10] (j);
\draw[line width=1pt,->]
(i) to (j);
\draw[line width=1pt,->]
(k) to (i);
\end{tikzpicture}
\caption{}
\label{fig: g}
\end{subfigure}
\begin{subfigure}{0.31\textwidth}
\begin{tikzpicture}
\tikzstyle{every node}=[draw, shape=circle];
\path (4, 0cm) node (j) {$j$};
\path (6.2, 0cm) node (k) {$k$};
\path (5.1, 1.6cm) node (i) {$i$};
\draw[line width=1pt,->] 
(j) to [bend left=10] (k);
\draw[line width=1pt,->] 
(k) to [bend left=10] (j);
\draw[line width=1pt,->] 
(i) to [bend left=10] (k);
\draw[line width=1pt,->] 
(k) to [bend left=10] (i);
\draw[line width=1pt,->]
(i) to (j);
\end{tikzpicture}
\caption{}
\label{fig: h}
\end{subfigure}
\label{fig: Uni-cyclic}
\caption{Unicyclic graphs of size three \citep{yoo2021new}}
\label{Unicyclic graphs of size 3}
\end{figure}
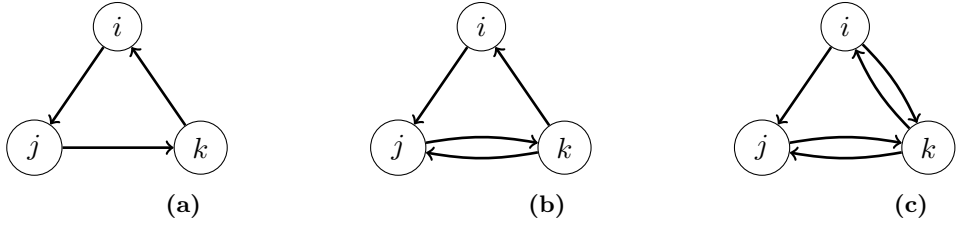
\begin{figure}[H]
\centering
\begin{subfigure}{0.225\textwidth}
\begin{tikzpicture}
\tikzstyle{every node}=[draw, shape=circle];
\path (0, 0cm) node (j) {$j$};
\path (2.2, 0cm) node (k) {$k$};
\path (1.1, 1.6cm) node (i) {$i$};
\draw[line width=1pt,->] 
(j) to (k);
\draw[line width=1pt,->] 
(i) to (j);
\draw[line width=1pt,->]
(i) to (k);
\end{tikzpicture}
\caption{}
\label{fig: a}
\end{subfigure}\hspace{3pt}
\begin{subfigure}{0.225\textwidth}
\begin{tikzpicture}
\tikzstyle{every node}=[draw, shape=circle];
\path (4, 0cm) node (j) {$j$};
\path (6.2, 0cm) node (k) {$k$};
\path (5.1, 1.6cm) node (i) {$i$};
\draw[line width=1pt,->] 
(j) to [bend left=10] (k);
\draw[line width=1pt,->] 
(k) to [bend left=10] (j);
\draw[line width=1pt,->]
(i) to (j);
\draw[line width=1pt,->]
(i) to (k);
\end{tikzpicture}
\caption{}
\label{fig: b}
\end{subfigure}\hspace{3pt}
\begin{subfigure}{0.225\textwidth}
\begin{tikzpicture}
\tikzstyle{every node}=[draw, shape=circle];
\path (4, 0cm) node (j) {$j$};
\path (6.2, 0cm) node (k) {$k$};
\path (5.1, 1.6cm) node (i) {$i$};
\draw[line width=1pt,->] 
(j) to [bend left=10] (k);
\draw[line width=1pt,->] 
(k) to [bend left=10] (j);
\draw[line width=1pt,->]
(j) to (i);
\draw[line width=1pt,->]
(k) to (i);
\end{tikzpicture}
\caption{}
\label{fig: c}
\end{subfigure}\hspace{3pt}
\begin{subfigure}{0.225\textwidth}
\begin{tikzpicture}
\tikzstyle{every node}=[draw, shape=circle];
\path (4, 0cm) node (j) {$j$};
\path (6.2, 0cm) node (k) {$k$};
\path (5.1, 1.6cm) node (i) {$i$};
\draw[line width=1pt,->] 
(j) to [bend left=10] (k);
\draw[line width=1pt,->] 
(k) to [bend left=10] (j);
\draw[line width=1pt,->]
(i) to [bend left=10] (j);
\draw[line width=1pt,->]
(i) to [bend left=10] (k);
\draw[line width=1pt,->]
(j) to [bend left=10] (i);
\draw[line width=1pt,->]
(k) to [bend left=10] (i);
\end{tikzpicture}
\caption{}
\label{fig: d}
\end{subfigure}
\caption{Unicycle-free graphs of size three \citep{yoo2021new}}
\label{Unicycle-free graphs of size 3}
\end{figure}
\begin{figure}[H]
\begin{small}
\centering
\hfill
\begin{subfigure}[b]{0.25\textwidth}
\begin{tikzpicture}
\tikzstyle{every node}=[draw, shape=circle];
\path (0, 0cm) node (1) {$1$};
\path (2, 0cm) node (2) {$2$};
\path (0, -2cm) node (3) {$3$};
\path (2, -2cm) node (4) {$4$};
\draw[line width=1pt,->] 
(1) to (3);
\draw[line width=1pt,->] 
(1) to (2);
\draw[line width=1pt,->]
(3) to (4);
\draw[line width=1pt,->]
(4) to (2);
\end{tikzpicture}
\end{subfigure}
\hfill
\begin{subfigure}[b]{0.25\textwidth}
\begin{tikzpicture}
\draw[blue] (1, 0cm) ellipse (1.55cm and 0.7cm);
\draw[blue] (1, -2cm) ellipse (1.55cm and 0.7cm);
\tikzstyle{every node}=[draw, shape=circle];
\path (0, 0cm) node (1) {$1$};
\path (2, 0cm) node (2) {$2$};
\path (0, -2cm) node (3) {$3$};
\path (2, -2cm) node (4) {$4$};
\draw[line width=1pt,green,->] 
(1) to (3);
\draw[line width=1pt,->] 
(1) to (2);
\draw[line width=1pt,->]
(3) to (4);
\draw[line width=1pt,red,->]
(4) to (2);
\end{tikzpicture}
\end{subfigure}
\hfill
\begin{subfigure}[b]{0.25\textwidth}
\begin{tikzpicture}
\draw[blue] (0, -1cm) ellipse (0.7cm and 1.55cm);
\draw[blue] (2, -1cm) ellipse (0.7cm and 1.55cm);
\tikzstyle{every node}=[draw, shape=circle];
\path (0, 0cm) node (1) {$1$};
\path (2, 0cm) node (2) {$2$};
\path (0, -2cm) node (3) {$3$};
\path (2, -2cm) node (4) {$4$};
\draw[line width=1pt,->] 
(1) to (3);
\draw[line width=1pt,green,->] 
(1) to (2);
\draw[line width=1pt,green,->]
(3) to (4);
\draw[line width=1pt,->]
(4) to (2);
\end{tikzpicture}
\label{fig: K-way resulting}
\end{subfigure}
\caption{A sample DAG and two partitions of its nodes}\llV
\label{a sample DAG and two partitions of its nodes}
\end{small}
\end{figure}
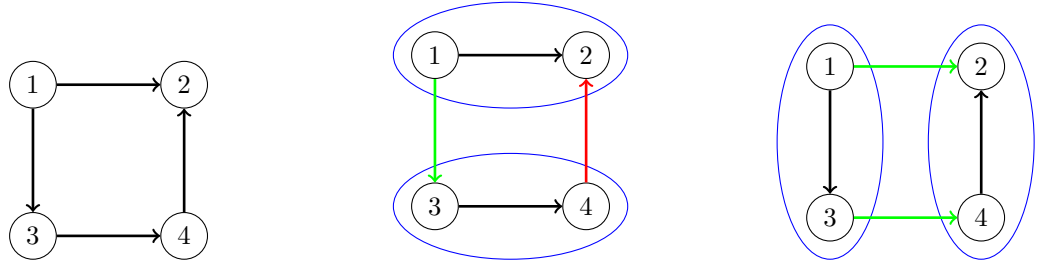

\section{A Generalization of the Feedback Arc Set Problem}\label{Sec: Generalized Feedback Arc Set Problem}\lV
This section introduces a new generalization of FASP, where unlike the traditional FASP, bidirectional arcs can be preserved. In particular, the traditional FASP can retain only one of the four types of unicycle-free subgraphs depicted in Fig. \ref{Unicycle-free graphs of size 3}, namely the one depicted in Fig. \ref{fig: a}, while the other three types, depicted in Figs. \ref{fig: b}-\ref{fig: d}, are forced to be eliminated by removing certain arcs, e.g., by removing either arc $(j,k)$ or arc $(k,j)$ in Fig. \ref{fig: b}. However, in the problem introduced herein, all types of of unicycle-free subgraphs can be retained: They may remain unchanged when keeping such cycles can help reflect the lack of a strict domination hierarchy, or they may be eliminated when doing so minimizes the weight of the removed arcs. However, as with the traditional FASP, all unicycles are forced to be eliminated. To solve the problem of interest, Section \ref{Weak Orderings and Unicycle-free Graphs Mapping} first reviews the relationship between DAGs and linear orders; then, it extends this relationship to the case of unicycle-free graphs and weak orderings. Furthermore, Section \ref{Bidirectional Feedback Arc Set Problem (BP-FASP)} formally defines BP-FASP, and Section \ref{Slater's Rule} establishes connections between BP-FASP and a concept from social choice theory known as Slater's rule.

\subsection{Mapping Weak Orderings and Unicycle-free Graphs}\label{Weak Orderings and Unicycle-free Graphs Mapping}
Every DAG can be associated with at least one linear ordering via a topological sorting algorithm. The associated linear ordering is guaranteed to be unique when there is full information on all pairwise relations, that is, when there is exactly one arc and/or directed path between each pair of nodes, as they determine a strict pairwise preference in the associated ordering; one example of such DAGs is a tournament graph (see Section Section \ref{Slater's Rule}). Consider the DAG shown in Fig. \ref{fig:resulting_graph}, which can be associated with the linear orderings $\boldsymbol{\sigma}^1 := \{1 \succ 2 \succ 3\succ 4\}$ and $\boldsymbol{\sigma}^2 := \{1 \succ 2 \succ 4\succ 3\}$. The graph admits both linear orderings because there is no directed arc or directed path between node-pair (3, 4). In general, a linear ordering $\boldsymbol{\sigma}$ can be associated with a DAG $\mathcal{G}$ as long as there are no directed arcs or directed paths in $\mathcal{G}$ from a lower-ranked node in $\boldsymbol{\sigma}$ to a higher-ranked node, that is, the ordering must not contradict the directionality of the arcs and simple paths of $\mathcal{G}$.
\begin{algorithm}[t]
\small
\SetKwInOut{Input}{Input}\SetKwInOut{Output}{Output}\SetKw{KwRet}{return}
\Input{unicycle-free graph $\mathcal{G} = (\mathcal{N}, \mathcal{A})$}
\Output{weak ordering $\boldsymbol{\sigma}$ associated with $\mathcal{G}$ (in its ordered partition representation)}
\nl $visited[v] = False \hspace{8pt} \forall v \in \mathcal{N}$\;
\nl Initialize the ordered partition representation: $\boldsymbol{\sigma} = \{\}$\;

\nl \While{$\exists v\in \mathcal{N}$ s.t. $visited[v] = False$}{
\nl select an arbitrary node $v$ where $visited[v] = False$\;
\nl Visit($v$)}
\nl $\boldsymbol{return} \boldsymbol{\sigma}$\;

\SetKwFunction{FSum}{Visit}
\SetKwProg{Fn}{Subroutine}{:}{}
\Fn{\FSum{$v$}}{
    \If{$visited[v] = True$}{\hspace{22pt} \KwRet{}}
    Let $\boldsymbol{\Xi} := \{u \in \mathcal{N}| (v, u) \in \mathcal{A}, \hspace{3pt} (u, v) \in \mathcal{A}\} \cup \{v\}$ be the subset of nodes that share a bidirectional arc with $v$, including $v$ itself\;
    % $\boldsymbol{\Xi} := \boldsymbol{\Xi} \cup \{v\}$\;
    Let $\boldsymbol{\Omega} := \{k \in \mathcal{N}| k \notin \boldsymbol{\Xi}, u \in \boldsymbol{\Xi}, \hspace{3pt}  (u, k) \in \mathcal{A}\}$ be the set of nodes with an incoming arc from any node in $\boldsymbol{\Xi}$\;
    \For{each node $k \in \Omega$}{Visit($k$)}
    \For{each node $u \in \Xi$}{$visited[u] = True$}

    % $visited[v] = True$\;
    Append $\boldsymbol{\Xi}$ to head of $\boldsymbol{\sigma}$\;}
\caption{Modified Topological Sorting Algorithm}
\label{Topological Sorting}
\end{algorithm}

Next, we present a modified topological sorting algorithm for associating a weak ordering with a given unicycle-free graph. The algorithm is an extension of DFS, and its pseudocode is provided in Algorithm \ref{Topological Sorting}. The algorithm starts by initializing each nodes as unvisited and the weak ordering as empty. While there are still unvisited nodes, the algorithm selects one such node $v$ at random, and then it invokes a recursive subroutine \emph{Visit($v$)}, which first checks whether $v$ has been visited. If true, it terminates the function; otherwise, it defines $\boldsymbol{\Xi}$ as the set of nodes that share a bidirectional arc with $v$, including itself. All nodes in $\boldsymbol{\Xi}$ will be tied together in the associated weak ordering. Then, it defines $\boldsymbol{\Omega}$ to be the set of all nodes not in $\boldsymbol{\Xi}$ that receive a directed arc from any node in $\boldsymbol{\Xi}$. That is, $\boldsymbol{\Omega}$ is the subset of nodes that will be ranked after $\boldsymbol{\Xi}$. In the next step, the algorithm invokes $Visit(k)$ for each $k\in\boldsymbol{\Omega}$. Once all nodes in $\boldsymbol{\Omega}$ have been visited, it marks all nodes in $\boldsymbol{\Xi}$ as visited as well, and it adds $\boldsymbol{\Xi}$ to the head of the ordered partition of the weak ordering. The algorithm terminates when all nodes have been visited. Similar to the original DFS topological sorting algorithm, Algorithm \ref{Topological Sorting} has a time complexity of $O\left(|V| + |A|\right)$. Note that the algorithm generates only one possible weak ordering for the given graph, even if multiple options exist. The random selection of an unvisited node in line four of Algorithm \ref{Topological Sorting} can yield different orderings in different runs. However, it can be de-randomized if all solutions are desired.

Consider the unicycle-free graph depicted in Fig. \ref{A sample unicycle-free graph with }. The following set of weak orderings can be associated with it through multiple runs of Algorithm \ref{Topological Sorting}: $\{\boldsymbol{\sigma^1} = \{2 \succ 4 \approx 5 \succ 3 \succ 1\}, \boldsymbol{\sigma^2} = \{2 \succ 4 \approx 5 \succ 1 \succ 3\}, \boldsymbol{\sigma^3} = \{2 \succ 1\succ 4 \approx 5 \succ 3\}\}$. The three weak orderings are consistent with the arcs present in the graph (the variation is due to the missing arcs/paths between certain nodes). In general, a weak ordering $\boldsymbol{\sigma}$ can be associated with a unicycle-free graph $\mathcal{G}$ as long as 1) there are no directed arcs or directed paths in $\mathcal{G}$ from a lower-ranked node to a higher-ranked node in $\boldsymbol{\sigma}$, and 2) there is a bidirectional arc in $\mathcal{G}$ between every pair of nodes that are tied in $\boldsymbol{\sigma}$. Notice that, whenever the unicycle-free graph $\mathcal{G}$ is a DAG, bidirectional arcs are absent, and the second condition is trivially satisfied.
\begin{figure}[ht]
\centering
\begin{tikzpicture}[scale=0.8]
\tikzstyle{every node}=[draw, shape=circle];
\path (10.4, 0cm) node (4) {$4$};
\path (13.4, 0cm) node (5) {$5$};
\path (11.9, 1.4cm) node (2) {$2$};
\path (9.2, 1.4cm) node (1) {$1$};
\path (14.6, 1.4cm) node (3) {$3$};
\draw[line width=1pt,->]
(4) to [bend left=10] (5);
\draw[line width=1pt,->]
(5) to [bend left=10] (4);
\draw[line width=1pt,->]
(2) to (1);
\draw[line width=1pt,->]
(5) to (3);
\draw[line width=1pt,->]
(2) to (4);
\end{tikzpicture}
\caption{A sample unicycle-free graph that can be associated with the weak orderings  $\boldsymbol{\sigma^1} = \{2 \succ 4 \approx 5 \succ 3 \succ 1\}, \boldsymbol{\sigma^2} = \{2 \succ 4 \approx 5 \succ 1 \succ 3\}$, and $\boldsymbol{\sigma^3} = \{2 \succ 1\succ 4 \approx 5 \succ 3\}\}$}
\label{A sample unicycle-free graph with }
\end{figure}
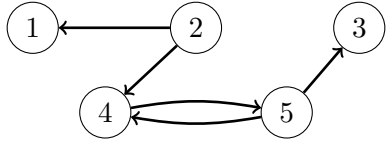

\allowdisplaybreaks
\subsection{The Bidirectionality-Preserving Feedback Arc Set Problem}\label{Bidirectional Feedback Arc Set Problem (BP-FASP)}
\begin{Definition}
(BP-FASP) Let $\mathcal{G} = (\mathcal{N}, \mathcal{A})$ be a weighted simple digraph. A bidirectionality-preserving feedback arc set of $\mathcal{G}$, denoted as $\mathcal{A}' \subseteq \mathcal{A}$, is a subset of arcs whose removal makes the resulting subgraph $\mathcal{G'} = (\mathcal{N}, \mathcal{A}\backslash \mathcal{A}')$ unicycle-free. The minimum bidirectionality-preserving feedback arc set of $\mathcal{G}$, abbreviated as MBP-FAS$(\mathcal{G})$, has the minimum total weight among all  bidirectionality-preserving feedback arc sets.
\end{Definition}

At first glance, the definitions of BP-FASP and FASP appear indistinguishable; in fact, since DAGs are a special case of unicycle-free graphs, FASP is consequently a special case of BP-FASP. The main difference between them is in the structure of the resulting graph $\mathcal{G}'$ (and by connection, the contents of $\mathcal{A}'$)\bla. Namely, with FASP, $\mathcal{G}'$ cannot contain bidirected arcs, but with BP-FASP certain bidirected arcs may be preserved. In particular, after the arc set $\mathcal{A}'$ is removed, the vertices of $\mathcal{G}'$ should form an acyclic k-way partition with an additional restriction: For each vertex subset $\mathcal{N}_{t} \in \boldsymbol{\mathcal{N}}:= \{\mathcal{N}_{1}, \dots, \mathcal{N}_{K}\}$, the subgraph induced by $\mathcal{N}_{t}$ must be a complete digraph. Note that FASP does not have an analogous requirement, as it does not retain bidirectional arcs (i.e., one of the two arcs must be removed). A resulting  distinguishing feature of the introduced problem is that the \textit{node subsets} $\mathcal{N}_{1}, \dots, \mathcal{N}_{K}$, rather than the individual nodes, form an acyclic graph. The following example helps illustrate these differences.
\begin{Example}
Consider the unweighted bidirectional graph depicted in Fig. \ref{subfig 1}, which has multiple cycles, e.g., $1\leadsto 3\leadsto 4\leadsto 1$. Here, MBP-FASP$(\mathcal{G})$ is shown in red, and the resulting unicycle-free subgraph and its corresponding 5-way partition of the nodes after removing MBP-FASP$(\mathcal{G})$ is shown in Fig. \ref{subfig 2}.
%\begin{small}
\begin{figure}[ht]
\centering
\resizebox{0.8\textwidth}{!}{
\begin{minipage}{\textwidth}
\begin{subfigure}[b]{0.5\textwidth}
\begin{tikzpicture}
\tikzstyle{every node}=[draw, shape=circle];
\path (0, 0cm) node (8) {$8$};
\path (1.6, 0cm) node (6) {$6$};
\path (0.8, 1.4cm) node (7) {$7$};
\path (3.2, 0cm) node (5) {$5$};
\path (4.8, 0cm) node (3) {$3$};
\path (4, 1.4cm) node (4) {$4$};
\path (2.4, 2.8cm) node (2) {$2$};
\path (5, 2.8cm) node (1) {$1$};
\draw[line width=1pt,->] 
(8) to [bend left=10] (6);
\draw[line width=1pt,->] 
(6) to [bend left=10] (8);
\draw[line width=1pt,->]
(7) to [bend left=10] (8);
\draw[line width=1pt,->]
(8) to [bend left=10] (7);
\draw[line width=1pt,->]
(6) to [bend left=10] (7);
\draw[line width=1pt,->]
(7) to [bend left=10] (6);
\draw[line width=1pt,red,->]
(3) to [bend left=10] (4);
\draw[line width=1pt,->]
(4) to [bend left=10] (3);
\draw[line width=1pt,->]
(3) to [bend left=10] (5);
\draw[line width=1pt,->]
(5) to [bend left=10] (3);
\draw[line width=1pt,->]
(4) to (5);
\draw[line width=1pt,red,->]
(6) to (4);
\draw[line width=1pt,->]
(1) to (3);
\draw[line width=1pt,->]
(4) to (1);
\draw[line width=1pt,->]
(1) to [bend left=10] (2);
\draw[line width=1pt,->]
(2) to [bend left=10] (1);
\draw[line width=1pt,->]
(2) to (6);
\draw[line width=1pt,->]
(2) to (7);
\draw[line width=1pt,red,->]
(5) to [bend left=10] (6);
\draw[line width=1pt,->]
(6) to [bend left=10] (5);
\end{tikzpicture}
\caption{}
\label{subfig 1}
\end{subfigure}
\hspace{15pt}
\begin{subfigure}[b]{0.5\textwidth}
\begin{tikzpicture}
\draw[blue] (0.8, 0.5cm) circle [radius=1.45];
\draw[blue] (2.4, 2.8cm) circle [radius=0.6];
\draw[blue] (5, 2.8cm) circle [radius=0.6];
\draw[blue] (4, 1.4cm) circle [radius=0.6];
\draw[blue] (4, 0cm) ellipse (1.5cm and 0.65cm);
\tikzstyle{every node}=[draw, shape=circle];
\path (0, 0cm) node (8) {$8$};
\path (1.6, 0cm) node (6) {$6$};
\path (0.8, 1.4cm) node (7) {$7$};
\path (3.2, 0cm) node (5) {$5$};
\path (4.8, 0cm) node (3) {$3$};
\path (4, 1.4cm) node (4) {$4$};
\path (2.4, 2.8cm) node (2) {$2$};
\path (5, 2.8cm) node (1) {$1$};
\draw[line width=1pt,->] 
(8) to [bend left=10] (6);
\draw[line width=1pt,->] 
(6) to [bend left=10] (8);
\draw[line width=1pt,->]
(7) to [bend left=10] (8);
\draw[line width=1pt,->]
(8) to [bend left=10] (7);
\draw[line width=1pt,->]
(6) to [bend left=10] (7);
\draw[line width=1pt,->]
(7) to [bend left=10] (6);
\draw[line width=1pt,->]
(4) to [bend left=10] (3);
\draw[line width=1pt,->]
(3) to [bend left=10] (5);
\draw[line width=1pt,->]
(5) to [bend left=10] (3);
\draw[line width=1pt,->]
(4) to (5);
\draw[line width=1pt,->]
(1) to (3);
\draw[line width=1pt,->]
(4) to (1);
\draw[line width=1pt,->]
(1) to [bend left=10] (2);
\draw[line width=1pt,->]
(2) to [bend left=10] (1);
\draw[line width=1pt,->]
(2) to (6);
\draw[line width=1pt,->]
(2) to (7);
\draw[line width=1pt,->]
(6) to [bend left=10] (5);
\end{tikzpicture}
\caption{}
\label{subfig 2}
\end{subfigure}
\\\
\caption{The MBP-FASP is comprised of arcs (3, 4), (5, 6) and (6, 4) (marked in red)}
\label{An example of the MBP-FASP comprised of arcs (3, 4), (5, 6) and (6, 4)}
\end{minipage}
}
\end{figure}
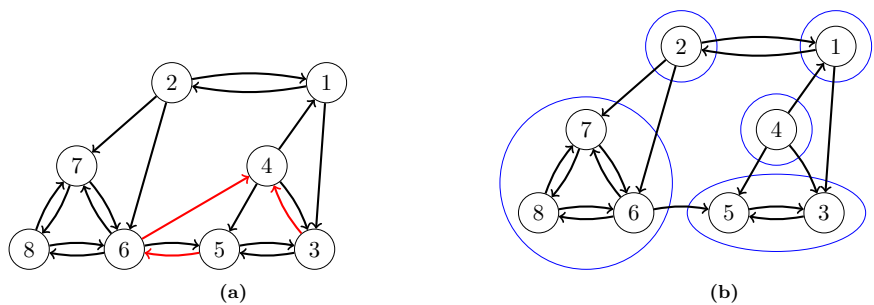\llV
%\end{small}
\label{Ex:BP-FASP}
\end{Example}
Example \ref{Ex:BP-FASP} demonstrates that all unicycles of $\mathcal{G}$ must be eliminated. The ensuing example demonstrates that arcs from complete subgraphs may at times also be removed, when doing so minimizes the cumulative weight of $\mathcal{A}'$. 
\begin{Example}
Consider the unweighted bidirectional graph depicted in Fig. \ref{fig 2: An example of the BP-FASP comprised of arcs (6, 1) and (6, 2)}, which has multiple cycles, e.g., $2 \leadsto 4 \leadsto 6 \leadsto 2$. Here, MBP-FASP$(\mathcal{G}) = \{(6, 1), (6, 2)\}$ is shown in red, and the resulting unicycle-free subgraph and its corresponding 5-way partition of the nodes after removing MBP-FASP$(\mathcal{G})$ is shown in Fig. \ref{fig 1: An example of the BP-FASP comprised of arcs (6, 1) and (6, 2)}. 
\begin{figure}[ht]
\centering
\resizebox{0.73\textwidth}{!}{
\hspace{-20pt}
\begin{minipage}{\textwidth}
\begin{subfigure}[b]{0.5\textwidth}
\begin{tikzpicture}
\tikzstyle{every node}=[draw, shape=circle];
\path (-0.5, 0cm) node (1) {$1$};
\path (3.5, 0cm) node (2) {$2$};
\path (-1.5, -1.75cm) node (3) {$3$};
\path (4.5, -1.75cm) node (4) {$4$};
\path (1.5, -1.75cm) node (5) {$5$};
\path (1.5, -3.5cm) node (6) {$6$};
\draw[line width=1pt,->]
(1) to [bend left=10] (2);
\draw[line width=1pt,->]
(2) to [bend left=10] (1);
\draw[line width=1pt,->]
(1) to [bend left=10] (6);
\draw[line width=1pt,red,->]
(6) to [bend left=10] (1);
\draw[line width=1pt,->]
(1) to (3);
\draw[line width=1pt,->]
(1) to (4);
\draw[line width=1pt,->]
(1) to (5);
\draw[line width=1pt,->]
(2) to [bend left=10] (6);
\draw[line width=1pt,red,->]
(6) to [bend left=10] (2);
\draw[line width=1pt,->]
(2) to (3);
\draw[line width=1pt,->]
(2) to (4);
\draw[line width=1pt,->]
(2) to (5);
\draw[line width=1pt,->]
(3) to (6);
\draw[line width=1pt,->]
(4) to (6);
\draw[line width=1pt,->]
(5) to (6);
\end{tikzpicture}
\caption{}
\label{fig 2: An example of the BP-FASP comprised of arcs (6, 1) and (6, 2)}
\end{subfigure}
\hspace{30pt}
\begin{subfigure}[b]{0.5\textwidth}
\begin{tikzpicture}
\draw[blue] (-1.5, -1.75cm) circle [radius=0.6];
\draw[blue] (4.5, -1.75cm) circle [radius=0.6];
\draw[blue] (1.5, -1.75cm) circle [radius=0.6];
\draw[blue] (1.5, -3.5cm) circle [radius=0.6];
\draw[blue] (1.5, 0cm) ellipse (3cm and 0.65cm);
\tikzstyle{every node}=[draw, shape=circle];
\path (-0.5, 0cm) node (1) {$1$};
\path (3.5, 0cm) node (2) {$2$};
\path (-1.5, -1.75cm) node (3) {$3$};
\path (4.5, -1.75cm) node (4) {$4$};
\path (1.5, -1.75cm) node (5) {$5$};
\path (1.5, -3.5cm) node (6) {$6$};
\draw[line width=1pt,->]
(1) to [bend left=10] (2);
\draw[line width=1pt,->]
(2) to [bend left=10] (1);
\draw[line width=1pt,->]
(1) to (6);
\draw[line width=1pt,->]
(1) to (3);
\draw[line width=1pt,->]
(1) to (4);
\draw[line width=1pt,->]
(1) to (5);
\draw[line width=1pt,->]
(2) to (6);
\draw[line width=1pt,->]
(2) to (3);
\draw[line width=1pt,->]
(2) to (4);
\draw[line width=1pt,->]
(2) to (5);
\draw[line width=1pt,->]
(3) to (6);
\draw[line width=1pt,->]
(4) to (6);
\draw[line width=1pt,->]
(5) to (6);
\end{tikzpicture}
\caption{}
\label{fig 1: An example of the BP-FASP comprised of arcs (6, 1) and (6, 2)}
\end{subfigure}
\caption{The MBP-FASP is comprised of arcs (6, 1) and (6, 2) (marked in red)}
\label{An example of the BP-FASP comprised of arcs (6, 1) and (6, 2)}
\end{minipage}
}
\end{figure}
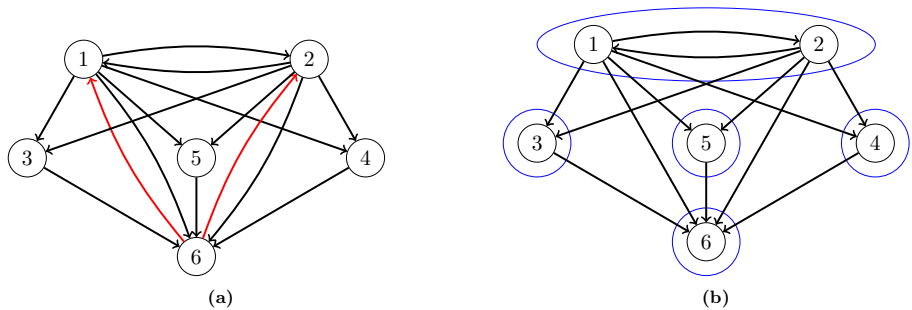
\end{Example}\lllV

\subsection{Slater's Rule Extension to Weak Tournaments}\label{Slater's Rule}
\noindent To elaborate on the connection between BP-FASP and Slater's Rule, we first introduce a few concepts from social choice theory. A \textit{tournament} $\cT=(\cV,\succeq)$ is defined as a complete and asymmetric binary relation $\succeq$ over a set of candidates $\cV$. Its elements are denoted equivalently as $(a,b)$ or $a\succeq b$, where $a,b\in\cV$---note that, because $(a,b)\in\cT$ implies $(b,a)\notin\cT$ (due to asymmetry), writing $a\succeq b$ implies $a\succ b$. Tournaments are essential across many fields including in social choice \citep{brandt2016handbook}, where they encapsulate the inputs required by a class of voting rules known as \textit{C1 social choice functions}. Although social choice functions generally assume that $p$ voters each provide a linear ordering over the set of candidates $\cV$, those belonging to the C1 class require only the majority relation conveyed through the collective votes between every pair of candidates, to determine a non-empty set of winners $S(\cT)\subseteq\cV$. Pairwise majority relations can be encoded with a tournament, therefore, in this context, $(a,b)\in\cT$ indicates that a strict majority of voters (i.e., $>p/2$) prefers $a$ over $b$. 

C1 functions, alternatively known as \textit{tournament solutions}, are differentiated by their axiomatic properties and the processes they apply to select $S(\cT)$ \citep{hudry2009survey}. Determining the winner(s) is complicated by the presence of \textit{pairwise-majority preference cycles}, which occur whenever $(a,b), (b,c),$ $(c,a)\in\cT$, for $a,b,c\in\cV$ with $a\neq b\neq c$. \textit{Slater's rule} is the tournament solution that finds the minimum number of majority relation reversals (i.e., replacements of $(a,b)$ with $(b,a)$) needed to remove all pairwise-majority preference cycles in $\cT$; it selects the maximal alternative from each topological (i.e., linear) ordering that can be obtained by making the minimum number of changes to $\cT$.

The existing literature on tournament solutions tends to focus on the case where there is a strict pairwise majority relation between every pair of alternatives in $\cV$; this is guaranteed by assuming that $p$ is odd. However, this assumption is rather artificial and difficult to justify in various practical social choice contexts \citep{brandt2018extending}. When it is removed, ties in the pairwise majority relations are possible and can be encoded with a \textit{weak tournament} $\cW=(\cV,\succeq)$, defined as a complete binary relation over $\cV$. For any pair of candidates $a,b\in\cV$, a strict pairwise majority relation $a\succ b$ is encoded as with conventional tournaments. Whenever there is no strict majority who prefers $a$ over $b$ (and vice versa), i.e., $a\approx b$, then both $(a,b), (b,a)$ are added as elements of $\cW$. Weak tournament \textit{extensions} or generalizations have been introduced for many C1 social choice functions (e.g., see \citep{per99con}). Formally, a weak tournament solution $S$ is an extension of a tournament solution $S'$ if $S(\cW)=S'(\cW)$, whenever $\cW\in\cT$ \citep{brandt2018extending}. To the best of our knowledge, no weak tournament extensions have been introduced specifically for Slater's rule. The ensuing theorem, whose proof is provided in Appendix \ref{Appendix Proof of Slater's Theorem}, introduces one such extension based on BP-FASP. 
\begin{Theorem}\label{Theorem: Slater}
Let $\cW$ define a weak tournament over a set of candidates $\cV$. Define a weak tournament solution $S$ as the alternative(s) in the maximal subset of the $K$-way partition obtained by making the minimum number of majority relation reversals to $\cW$. Then, $S$ is an extension of Slater's rule to weak tournaments. 
\end{Theorem}
Because FASP is a special case of BP-FASP, the latter inherits the computational complexity of the former. This includes being NP-hard \citep{karp1972reducibility} and also APX-hard \citep{kann1992approximability} on general graphs, meaning that there is no polynomial-time approximation scheme (PTAS) for FASP (and BP-FASP), unless P = NP. Note, however, that FASP has a PTAS for tournament graphs \citep{kenyon2007rank} and planar graphs \citep{grotschel1985acyclic}. We conjecture that, under similar graph restrictions, BP-FASP (and thus the featured Slater's rule extension) has a PTAS. 

It is important to explain that, although a generic generalization known as the \textit{conservative extension} can be applied to any C1 social choice function \citep{brandt2016handbook}, its implementation for Slater's rule would prove impractical. The conservative extension works by generating all possible (strict) tournament orientations embedded within a weak tournament, applying the given tournament solution to each orientation in the usual way, and finally including the winners from all of the tournament orientations as the winners of the weak tournament. For instance, if $\cW$ contains only one tie in the pairwise majority relations, say for pair $a$ and $b$, then orientation $\cT_{a\succ b}\in\cW$ would contain $(a,b)$ but not $(b,a)$, and orientation $\cT_{b\succ a}\in\cW$ would contain $(b,a)$ but not $(a,b)$. The winners of $\cW$ according to tournament solution $S$ would be given by $S[\cW]= S(\cT_{b\succ a}) \cup S(\cT_{a\succ b})$. Two ties in the pairwise majority relations yield four possible tournament orientations and, in general, with $q$ ties, there are $2^q$ possible tournament orientations, where $q\le n(n-1)/2$ \citep{brandt2018extending}. This creates two types of challenges for the case of Slater's rule. From a social choice perspective, the set of winners may become overly large since Slater's rule may yield multiple alternative optimal solutions when applied to a single tournament orientation. From a computational perspective, an instance of FASP must be solved for each orientation, making this approach prohibitive, even for a small number of pairwise majority ties. Conversely, BP-FASP entails solving one NP-hard problem. It is also reasonable to conjecture that the size of winning set with BP-FASP would be smaller than with the conservation extension. We leave the exploration of this conjecture and the exploration of a PTAS for BP-FASP for future work.\llV

\section{Formulations and Polyhedral 
Comparison}\label{Sec: Mathematical Formulations}
This section proposes two binary linear programming (BLP) formulations for BP-FASP, and it uses polyhedral analysis to compare them. It is important to note that the proposed formulations directly seek a weak ordering that can be associated with the unicycle-free subgraph that results after $\mathcal{A}'$ is removed; hence, the contents of $\mathcal{A}'$ are  obtained in post-processing. Computing the traditional feedback arc set via mathematical optimization follows a similar process \citep{baharev2021exact}, regardless of whether the associated linear ordering is the primary object of interest.

Next, we discuss the motivation and underlying logic of the proposed formulations. Recall that there are two requirements for associating a weak ordering $\boldsymbol{\sigma}$ to a unicycle-free subgraph $\mathcal{G}'$ of $\mathcal{G}$: 1) there are no directed arcs or directed paths in $\mathcal{G}'$ from a lower-ranked node to a higher-ranked node in $\boldsymbol{\sigma}$, and 2) there is a bidirectional arc in $\mathcal{G}'$ between every pair of nodes that are tied in $\boldsymbol{\sigma}$. These two conditions shape the objective function and constraints of the mathematical formulations, respectively. In greater detail, combining the second requirement with the fact that $\mathcal{G}'$ is a subgraph of $\mathcal{G}$, if there are no bidirectional arcs between $i$ and $j$ in $\mathcal{G}$, then weak orderings that tie these two cannot be associated with $\mathcal{G}'$. That is, two nodes \textit{can} be tied in $\boldsymbol{\sigma}$ only if there is a bidirectional arc between them in $\mathcal{G}$. Nonetheless, one may be ultimately ranked ahead of the other, if doing so minimizes the objection function. Hence, there must be a strict ordering between any pair of nodes without a bidirectional arc. Moreover, due to first requirement, any arc $(i, j) \in \mathcal{A}$ for which $j$ is strictly ranked above $i$ in $\boldsymbol{\sigma}$ must be removed in forming $\mathcal{G}'$. Hence, the objective function minimizes the summation of the weights of such removed arcs.

Let $y_{ij}$ be equal to 1 if $i$ is \textit{ranked ahead of or tied with} $j$ in $\boldsymbol{\sigma}$, and 0 otherwise. Two nodes are tied in $\boldsymbol{\sigma}$ if $y_{ij} = y_{ji} = 1$, and $i$ is strictly ranked ahead of $j$ if $y_{ij} = 1, y_{ji} = 0$. The first proposed formulation, denoted as BLP \#1, is given by:\llV
\allowdisplaybreaks
{\small
\begin{subequations}\label{prob: BMFARS}
\begin{align}
\text{(BLP \#1)} \quad \min \quad & z = \sum_{(i, j) \in \mathcal{A}} w_{ij}(1 - y_{ij}) \label{eq: Formulation1-OBJ} \\
\text{s.t.} \quad 
& y_{ij} + y_{ji} \geq 1 && \hspace{1pt} \forall i, j \in \mathcal{N}|\hspace{2pt} (i, j) \in \mathcal{A} \hspace{5pt}\text{and}\hspace{5pt} (j, i) \in \mathcal{A} \label{eq: Formulation1-eq1}\\
& y_{ij} + y_{ji} = 1 && \hspace{1pt} \forall i, j \in \mathcal{N}|\hspace{2pt} (i, j) \notin \mathcal{A} \hspace{5pt}\text{or}\hspace{5pt} (j, i) \notin \mathcal{A} \label{eq: Formulation1-eq2}\\
& y_{ij} - y_{ik} - y_{kj} \geq -1 &&\hspace{1pt} \forall i, j, k \in \mathcal{N}; \hspace{2pt} i \neq j \neq k  \label{eq: Formulation1-eq3}
% & y_{ij} \in \{0, 1\} && \hspace{1pt}  \forall i, j \in \mathcal{N}; \hspace{2pt}  i \neq j\label{eq: Formulation1-eq4}
\end{align}
\end{subequations}
}
Constraints \eqref{eq: Formulation1-eq1}-\eqref{eq: Formulation1-eq2} determine the relative ordering of each item-pair in $\boldsymbol{\sigma}$. That is, $i$ and $j$ \emph{may} be tied in $\boldsymbol{\sigma}$ only if there is a bidirectional arc between them in $\mathcal{G}$; otherwise, one must be ranked ahead of the other. Constraint \eqref{eq: Formulation1-eq3} enforces (preference) transitivity in $\boldsymbol{\sigma}$ by preventing unicycles of size three; we refer the reader to \cite{yoo2021new} for more details. Given a feasible solution to BLP \#1, $\mathcal{A}'$ is obtained as the set of arcs $(i, j) \in \mathcal{A}$ where $j$ is strictly ranked ahead of $i$, that is, $\mathcal{A}' := \{(i, j) \in \mathcal{A}|y_{ij} = 0\}$ (these are the arcs whose weights correspond to the minimum objective function value). The rank of node $i$ in $\boldsymbol{\sigma}$ is obtained as $n - \sum_{j \in \mathcal{N}: i \neq j} y_{ij}$. 
\begin{Proposition}\label{Proposition: admit} 
BLP \#1 admits any weak ordering that can be associated with a unicycle-free subgraph within a given digraph $\mathcal{G} = (\mathcal{N}, \mathcal{A})$.
\end{Proposition}
The proof of Proposition \ref{Proposition: admit} is provided in Appendix \ref{Appendix Proof of Proposition 1}.

\FloatBarrier
In BLP \#1, the transitivity of $\boldsymbol{\sigma}$ is guaranteed via Constraint \eqref{eq: Formulation1-eq3}, which have cardinality of $O(n^3)$. Next, we introduce and describe an alternative set of constraints, with reduced cardinality of $O(n^2)$, which are given by:\vspace{-10pt}
\allowdisplaybreaks
\begin{small}
\begin{subequations}\label{eq: Other Transitivity}
\begin{align}
\allowdisplaybreaks
&\sum_{u \in \mathcal{N}\backslash \{i, j\}}(y_{uj} - y_{ui}) \geq (n - 2)(y_{ij} - 1), \hspace{6pt} \forall i, j \in \mathcal{N}; i \neq j, \label{eq: 161}\\
&\sum_{u \in \mathcal{N}\backslash \{i, j\}}(y_{iu} - y_{ju}) \geq (n - 2)(y_{ij} - 1), \hspace{6pt} \forall i, j \in \mathcal{N}; i \neq j,\label{eq: 162}\\
&\sum_{u \in \mathcal{N}\backslash \{i, j\}}(y_{iu} + y_{uj}) \leq (n - 2)(y_{ij} + 1), \hspace{6pt} \forall i, j \in \mathcal{N}; i \neq j.\label{eq: 163}
\end{align}
\end{subequations}   
\end{small}
Constraint \eqref{eq: 161} states that if $i$ and $j$ are tied or if $i$ is ranked ahead of $j$, i.e. $y_{ij} = 1$, the number of items that are ranked ahead of $j$ must be greater than or equal to the number of items that are ranked ahead of $i$; it becomes redundant otherwise. Constraint \eqref{eq: 162} states that if $i$ and $j$ are tied or if $i$ is ranked ahead of $j$, the number of items over which $i$ is ranked must be greater than or equal to the number of items over which $j$ is ranked. Constraint \eqref{eq: 163} states that if $j$ is strictly ranked ahead of $i$, i.e. $y_{ij} = 1$, the total number of items that are ranked ahead of $j$, plus the number of items over which $i$ is ranked ahead, must less than or equal to $n-2$ (the total remaining items). Note that Constraints \eqref{eq: Other Transitivity} can be applied to any mathematical formulation concerned with the transitivity of rankings, such as in LOP and its applications, including ranking aggregation. The following proposition, whose proof is provided in Appendix \ref{Appendix Proof of transitivity}, proves that these constraints provide an alternative to Constraint \eqref{eq: Formulation1-eq3} for enforcing transitivity.
\begin{Proposition}\label{prop: transitivity}
The transitivity requirement can be imposed by Constraints \eqref{eq: Other Transitivity}.
\end{Proposition}
Next, we introduce a second formulation, denoted as BLP \#2, which is obtained by replacing Constraint \eqref{eq: Formulation1-eq3} by Constraints \eqref{eq: Other Transitivity}, and is given by $\{\underset{y}{\min} \hspace{5pt} \eqref{eq: Formulation1-OBJ} \hspace{5pt} \text{s.t.} \hspace{5pt}\eqref{eq: Formulation1-eq1}-\eqref{eq: Formulation1-eq2}, \eqref{eq: Other Transitivity}\}$.

The ensuing theorem compares the strength of the proposed BLPs by analyzing their linear programming (LP) polyhedra. The proof is given in Appendix \ref{Appendix Proof of Polyhedral Theorem}. 
\begin{Theorem}\label{Theorem: polydehral}
Let $\mathcal{P}^1$ and $\mathcal{P}^2$ denote the LP-relaxed regions of BLPs \#1 and \#2, respectively. For any instance of BP-FASP, $\mathcal{P}^1 \subseteq \mathcal{P}^2$, and this inclusion can be strict.
\end{Theorem}
Theorem \ref{Theorem: polydehral} indicates that the LP-relaxation of BLP \#1 gives an objective value that is always at least as good as that of BLP \#2, and sometimes better. This means BLP \#1 offers a lower bound that is never worse than BLP \#2. However, solving the LP-relaxation of BLP \#1 requires significantly more time than BLP \#2, due to its higher number of constraints. This is unsurprising, as the performance of LP solution methods depends on both the number of variables and constraints. In preliminary tests, BLP \#2 obtained  competitive lower bounds on randomly generated instances in far less time. Future work will extend this analysis in directions similar to those of \citep{conitzer2006improved} and \citep{akbari2021lower}, who study the computational time and solution quality of different lower-bounding techniques for strict and non-strict rank aggregation problems, respectively.\llV\lV 

\section{IEOR Faculty Hiring Network}\label{Sec: Case Study}
This section studies the IEOR FHN \citep{del2020exponential}. Section \ref{University Rankings and Faculty Hiring Networks} provides an overview of university rankings, FHNs, and the methods of analyzing FHNs. Section \ref{IEOR Dataset Description} describes the IEOR dataset, and Section \ref{Results} compares BP-FASP methodology versus existing approaches to determine the ranking and tiering of IEOR departments. 
\allowdisplaybreaks\llV
\subsection{University Rankings and Faculty Hiring Networks}\label{University Rankings and Faculty Hiring Networks}
University rankings have wide-ranging implications. They impact prospective students who use them to identify suitable universities to attend. Additionally, they impact the status of higher education institutions and their external funding, as rankings are said to reflect educational excellence. Naturally, a higher rank attracts high-quality scholars and more funding, including donations, gifts, and endorsements. Many outlets produce university ranking lists including U.S.NEWS, the Quacquarelli Symonds (QS), and the Times Higher Education (THE). Among them, U.S.NEWS is unique in its publication of discipline-specific university rankings. 

The methodologies used by these organizations raise several important concerns. First, there is ongoing debate over the selection of criteria, their associated weights, and their interdependencies \citep{ccakir2015comparative}. Take, for instance, U.S.NEWS's methodology for ranking graduate engineering schools, which is based on four main criteria: research activity (50\%), quality assessment (25\%), faculty resources (20\%), and student selectivity, i.e., acceptance rate (5\%). Each of these comprises multiple sub-criteria, for example, research activity is made up of total research expenditures (25\%)---a total of externally funded, public or private, engineering research expenditures---average research expenditures per faculty member (15\%), and faculty research (10\%)---assessed through bibliometric publication and citation data \citep{USNEWSMethodology}. The process of selecting criteria and weights used by these institutions, along with their ranking methodology, is effectively a black-box in that it can be difficult to access the source data and replicate results. Second, educational inputs such as research funds and selectivity may outweigh educational output. Third, these methodologies may rely on surveys and self-reported data, which can be susceptible to intentional and\textbackslash or accidental errors and bias. To highlight this drawback, consider a recent scandal that caused Columbia University to be dropped from rank 2 to 18 in the U.S.NEWS Best National Universities list due to the submission of inaccurate data by university officials \citep{Columbia}. Since then, many prominent programs have opted out of certain popular rankings including the Law and Medical Schools of Harvard University, Yale University, Columbia University, and Stanford University \citep{bestcolleges}.

The inherent drawbacks of rankings methodologies have prompted researchers to look for remedies. A common suggestion for mitigating the risks of relying on individual rankings is to aggregate the outputs from multiple organizations into a single \emph{consensus ranking} \citep{zhang2021comprehensive}. However, this approach has other drawbacks including inconsistencies in the programs/universities and time horizon evaluated and inaccessibility of datasets. The limitations of using these privatized sources can be sidestepped altogether through outcome-driven approaches based on publicly available data. Chief among these is the analysis of FHN. In effect, when an institution hires a Ph.D. graduate from another institution as a tenure-track faculty, the former effectively bestows prestige to the latter. Thus, an institution's dual abilities to place/attract Ph.D. graduates at/from other programs can determine its position in the hierarchy. By recording and analyzing these hiring data, the academic departments and faculty themselves provide a collective implicit assessment of programs within a specific field of study.

FHNs of many programs have been studied including IEOR \citep{del2020exponential}, Business, Computer Science, History \citep{clauset2015systematic}, Law \citep{katz2011reproduction}, Mathematics \citep{fitzgerald2023temporal}, and Physics \citep{grunspan2024importance}. Study of FHNs can also help identify inequities in the faculty hiring process \citep{way2016gender} and reveal other sociological aspects \citep{katz2011reproduction}. \citet{orland2022there} examined the gender hiring gap in Economics FHN across Austria, Germany, and Switzerland, finding no significant gap. They noted that hired faculty, on average, step down 22.6 ranks from their doctoral institutions to their employing institutions. Similarly, faculties of business, computer science, and history in the U.S. and Canada step down, on average, between 27 and 47 ranks \citep{clauset2015systematic}. Additionally, many studies show that most faculty hires come from a few institutions; for example, only 10\% of IEOR departments trained nearly half of all faculty members \cite{del2020exponential}. As such, some argue that, due to the skewed placement power of institutions, the demographics and research agendas of a field are shaped by highly-ranked Ph.D. programs \citep{clauset2015systematic, way2016gender}. 

A widely used technique for analyzing FHNs is via the \emph{hub score} and \emph{authority score}, which are interconnected scores that belong to the more general \emph{centrality} scores in social networks \citep{katz2011reproduction}. \emph{Hubs} are nodes with many faculty placements, and \emph{authorities} are nodes that hire from such nodes. In other words, hub and authority scores measure the placement capacity and hiring capacity of departments, respectively. The FHNs of Mathematics \citep{fitzgerald2023temporal}, Physics \citep{grunspan2024importance}, and Law \citep{katz2011reproduction} have been analyzed using these concepts. 
% at other high-prestige nodes

Another method of analyzing FHN for deriving a prestige hierarchy is the \textit{minimum violation ranking} (MVR) problem and its variants. MVR measures prestige by finding a ranking of departments with the least number of “violating” arcs, i.e., the number of upward-pointing arcs $(j, i)$, where the rank of department $i$ is better than the rank of department $j$ in the hierarchy. That is, a “violation” is defined as placements of hires into more prestigious departments than their Ph.D.-granting institutions. In general, MVR seeks to rank a set of objects given a set of pairwise comparisons, i.e., akin to ordering teams based on the outcomes of pairwise matchups. MVR is equivalent to FASP \citep{ali1986minimum}, as such, MVR yields an inherently strict ranking.\lllV
\subsection{Description of the Dataset}\label{IEOR Dataset Description}\llV
The IEOR FHN dataset collected in May 2016 by \citet{del2020exponential}. The dataset considers 83 IEOR departments, which was formed by merging departments with a Ph.D. program in the 2016 U.S.NEWS ``Industrial/Manufacturing/Systems Engineering'' graduate rankings with those in the 2011 National Research Council rankings for ``Operations Research, Systems Engineering, and Industrial Engineering''. Only departments with web pages at the time of data collection were included, which excluded the Univ. of Nebraska-Lincoln. For those IEOR departments that are merged with other disciplines in a single department, only those faculty in IEOR field were included. Whenever the departmental information of a Ph.D. was not provided, it was assumed that the degree was granted from IEOR department of the listed institution. Faculty who did not receive their Ph.D. degree in the U.S. were excluded from the dataset, resulting in 1,179 tenured and tenure-track faculty to be included in the final computations. We refer the readers to \citet{del2020exponential} for more details on the dataset.
% (e.g., Mechanical and Industrial Engineering)

The arcs of the network\textbackslash graph are formed from faculties currently in institution $j$, as of May 2016, who received their Ph.D. from institution $i$; the weight of such arc $(i, j)$ equals the number of placements. This implies that the possible intermediate movement of faculties between institutions is ignored, due to a lack of information. Moreover, a \emph{self-hire} is defined as a case where an IEOR department hires a faculty who graduated from the same university, but not necessarily from itself. Only 47.5\% of self-hires were immediate, meaning that certain faculty members commenced their careers from a different institution but later in their career, they rejoined their alma mater. The IEOR graph is fairly sparse, having only 5\% of the total possible arcs (i.e., potential hiring interactions). \citet{del2020exponential} use steepness statistics to analyze the IEOR social network, finding a statistically significant near-linear hierarchy. However, the hierarchy appears relatively flat across most departments—primarily due to the network’s sparsity—except among the top-ranked departments, where it is more pronounced. An induced subgraph of the resulting IEOR FHN is depicted in Fig. \ref{fig:enter-label}.\mmmV %in the online companion.

\begin{figure}[ht]
    \centering
    \vspace{-16pt}
    \includegraphics[scale=0.28]{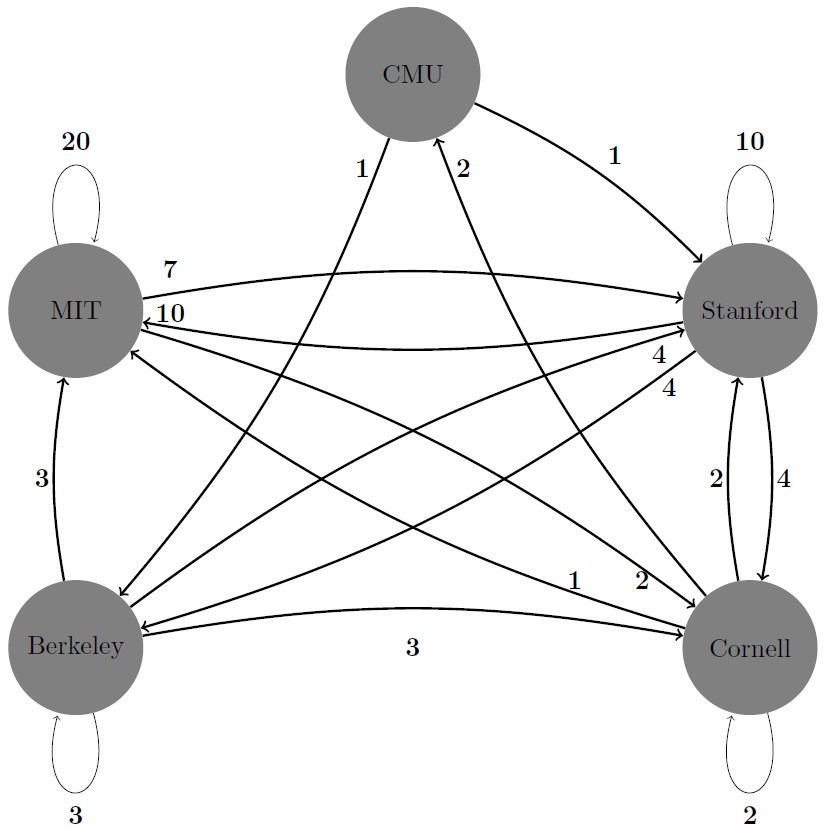}
    \caption{The induced graph reveals the high degree of mutual hiring relationships among Massachusetts Institute of Technology (MIT), Stanford Univ. (Stanford), Univ. of California Berkeley (Berkeley), Carnegie Melon Univ. (CMU), and Cornell Univ. (Cornell)}\vspace{-0.1in}
    \label{fig:enter-label}
\end{figure}\lllV

\subsection{Results}\label{Results}\lV
To analyze the IEOR FHN, we first review the methods used by \citet{del2020exponential}, who deploy MVR and two variants thereof that take into account the \emph{magnitude of violation} and \emph{unexpected placements}\bla. The first variant, minimum Violation and Strength or MVS$_1$, takes into account the difference between the two departments' rank values (i.e., the magnitude of the ranking violations). For example, the penalty of placing a graduate in the first-ranked school from the 11th-ranked school is five times the penalty of the same placement coming from the second-ranked school. The second variant, MVS$_2$, builds on MVS$_1$ and additionally accounts for the situation where a top-ranked department places too many faculties in a lower-ranked department. This outcome is considered unattractive, because it is reasonable to expect that a department with a high number of placements from top-ranked departments should be ranked ahead of departments that fail both to place their Ph.D. graduates at and attract graduates from the top-ranked departments. This corresponds to the case where two departments have closely low hub scores, but one has a higher authority score and should rank higher \citep{del2020exponential}.
% This corresponds to the case where two departments have closely low hub scores but one has a higher authority score and, thus, should be ranked higher \citep{del2020exponential}. 

Due to the sparseness of the network, \citet{del2020exponential} consider the digraph constructed from the input data as a noisy realization (i.e., a sample) of a stochastic network model. To obtain additional observations, Markov Chain Monte Carlo (MCMC) is utilized to sample edges with replacement with a probability of selecting an edge being proportional to its weight. Their computational study generates 1,000 such samples and applies MVR, MVS$_1$, MVS$_2$ to each sampled network. The final rankings yielded by each method are obtained by sorting the departments based on their the average ranks over all the bootstrapped networks, inserting ties when the average ranks of a subset of departments are not statistically significant. Note that MVR and its variants cannot account for ties and may require potentially  computationally intensive steps to incorporate them.

The resulting rankings are displayed in Table \ref{Rankings Table Aggregated} alongside those obtained by FASP, BP-FASP, U.S.NEWS, HUB, and AUTHORITY. Note that the FASP and MVR rankings differ because, as previously explained, the latter are obtained in \cite{del2020exponential} by sorting the departments based on the average ranks over 1,000 bootstrapped networks, whereas we obtain the FASP ranking by directly applying it to the original network induced by the given dataset. Furthermore, in post-processing, we introduce additional ties to BP-FASP to incorporate a notion of fairness. Specifically, a subset of adjacent universities in BP-FASP is considered tied if they have no mutual hiring and any permutation of these universities yields the same objective function value. These ties, which differ from those arising from equivalence in prestige, are marked with an asterisk in Table \ref{Rankings Table Aggregated}. Additionally, since FASP and BP-FASP are defined over simple digraphs, self-hires are effectively ignored in FASP, BP-FASP, MVR, MVS$_1$, and MVS$_2$ rankings. Self-hires are discussed in more detail later in this section.

The fine-grained ranks from the department rankings fail to capture hierarchical information, for instance, the existence of a group characterized by both frequent, mutual hires within the group and infrequent interactions with other institutions. To unveil such insights, \citet{del2020exponential} cluster the departments into hierarchies using \emph{exponential random graph modeling (ERGM)}. This approach collapses the network, allowing for the modeling of placements and hires within groups. ERGM is a statistical distribution on network data, characterizing the distribution in terms of a set of summary statistics of the network, usually calculated from the network's topological features, such as the number of arcs and subgraph counts \citep{lehmann2024bayesian}. \citet{del2020exponential} propose a new latent variable ERGM model in which the conditional probability of an arc, given covariates, depends on the distance of its incident nodes in a latent space. The authors use the MVS$_2$ ranking as vertex covariates, which explain the formation of edges, in the latent variable ERGM. Their clustering divides the IEOR departments into three tiers: 14 departments are in Tier 1 (top-tier), 36 in Tier 2 (mid-tier), and 34 in Tier 3 (low-tier).\mmV

\begin{table}[h!]
\vspace{-11pt}
\caption{Ranking and tiering of IEOR departments using various metrics}
\centering
%\Huge
\Gigantic{54pt}
\resizebox{\columnwidth}{!}{
\begin{tabular}{|p{20cm}|c|c|c|c|c|c|c|c|c|c|}
\hline
\multirow{2}{*}{\textbf{University}} & \multicolumn{8}{c|}{\textbf{Rank}} & \multicolumn{2}{c|}{\textbf{Tier}}\\ \cline{2-11} & \textbf{BP-FASP} & \textbf{FASP} & \textbf{MVR} & \textbf{MVS} & \textbf{MVS$_2$} & \textbf{HUB} & \textbf{AUTHORITY} & \textbf{U.S.NEWS} & \textbf{D-C} & \textbf{BP-FASP}\\ \hline
Carnegie Mellon University & 1 & 1 & 4 & 4 & 4 & 8 & 59 & N/A & 1 & 1\\ \hline
University of California-Berkeley & 2 & 2 & 1 & 1 & 2 & 3 & 18 & 2 & 1 & 1\\ \hline
Massachusetts Institute of Technology & 3 & 4 & 3 & 3 & 3 & 1 & 1 & 6 & 1 & 1\\ \hline
Stanford University & 3 & 3 & 2 & 2 & 1 & 2 & 2 & 4 & 1 & 1\\ \hline
Cornell University & 5 & 7 & 7 & 7 & 7 & 7 & 12 & 7 & 1 & 1\\ \hline
Princeton University & 5 & 5 & 5 & 5 & 5 & 15 & 16 & N/A & 1 & 1\\ \hline
Columbia University & 7 & 8 & 8 & 8 & 8 & 17 & 6 & 11 & 1 & 1\\ \hline
\textbf{University of Michigan Ann Arbor} & 8 & 9 & 6 & 6 & 6 & 4 & 8 & 2 & \bf{2} & \bf{1}\\ \hline
Northwestern University & 9 & 10 & 11 & 11 & 11 & 10 & 28 & 4 & 1 & 1\\ \hline
\textbf{Georgia Institute of Technology} & 10 & 12 & 4 & 10 & 10 & 5 & 3 & 1 & \bf{2} & \bf{1}\\ \hline
\textbf{Pennsylvania State University} & 10 & 15 & 15 & 16 & 16 & 13 & 15 & 12 & \bf{2} & \bf{1}\\ \hline
\textbf{Purdue University West Lafayette} & 10 & 11 & 9 & 9 & 9 & 6 & 13 & 9 & \bf{2} & \bf{1}\\ \hline
\textbf{University of Illinois Urbana-Champaign} & 10 & 13 & 14 & 14 & 12 & 14 & 5 & 15 & \bf{2} & \bf{1}\\ \hline
Lehigh University & 14 & 18 & 26 & 23 & 27 & 43 & 40 & 18 & 2 & 2\\ \hline
\textbf{University of Texas Austin} & 14 & 6 & 16 & 18 & 19 & 22 & 21 & 19 & \bf{1} & \bf{2}\\ \hline
University of Minnesota Twin Cities & 16 & 16 & 20 & 20 & 20 & 29 & 31 & 32 & 2 & 2\\ \hline
Ohio State University & \hspace{4pt} 17$^{*}$ & 23 & 21 & 21 & 18 & 19 & 19 & 17 & 2 & 2\\ \hline
University of Florida & \hspace{4pt} 17$^{*}$ & 17 & 13 & 13 & 15 & 16 & 38 & 19 & 2 & 2\\ \hline
\textbf{University of Iowa} & 19 & 24 & 22 & 22 & 22 & 28 & 71 & 39 & \bf{3} & \bf{2}\\ \hline
University of Pittsburgh & 20 & 25 & 24 & 24 & 23 & 24 & 20 & 23 & 2 & 2\\ \hline
University of Wisconsin Madison & 21 & 14 & 12 & 12 & 13 & 9 & 22 & 7 & 2 & 2\\ \hline
\textbf{University of Pennsylvania} & 22 & 19 & 18 & 15 & 14 & 20 & 7 & 28 & \bf{1} & \bf{2}\\ \hline
North Carolina State University & 23 & 26 & 32 & 33 & 26 & 30 & 29 & 12 & 2 & 2\\ \hline
Virginia Tech & 23 & 27 & 25 & 25 & 24 & 12 & 25 & 9 & 2 & 2\\ \hline
Rutgers University & \hspace{4pt} 25$^{*}$ & 29 & 30 & 28 & 30 & 44 & 49 & 21 & 2 & 2\\ \hline
University of Missouri Columbia & \hspace{4pt} 25$^{*}$& 20 & 23 & 26 & 29 & 45 & 63 & 58 & 2 & 2\\ \hline
\textbf{University of Southern California} & \hspace{4pt} 25$^{*}$ & 28 & 19 & 19 & 17 & 25 & 10 & 12 & \bf{1} & \bf{2}\\ \hline
University of Maryland College Park & 28 & 21 & 17 & 17 & 21 & 11 & 32 & N/A & 2 & 2\\ \hline
Boston University & 29 & 22 & 28 & 30 & 35 & 53 & 24 & 39 & 2 & 2\\ \hline
University of Arizona & \hspace{4pt} 30$^{*}$ & 33 & 27 & 27 & 25 & 27 & 51 & 28 & 2 & 2\\ \hline
\textbf{University of North Carolina Chapel Hill} & \hspace{4pt} 30$^{*}$ & 31 & 42 & 35 & 31 & 38 & 36 & N/A & \bf{1} & \bf{2}\\ \hline
Rensselaer Polytechnic Institute & 32 & 32 & 37 & 41 & 39 & 39 & 45 & 21 & 2 & 2\\ \hline
University of Arkansas Fayetteville & 32 & 34 & 41 & 40 & 36 & 49 & 43 & 39 & 2 & 2\\ \hline
\textbf{University of Illinois Chicago} & 34 & 30 & 33 & 36 & 42 & 48 & 61 & 46 & \bf{3} & \bf{2}\\ \hline
Texas A\&M University & 35 & 36 & 34 & 31 & 32 & 23 & 30 & 15 & 2 & 2\\ \hline
\textbf{University of South Florida} & 35 & 37 & 36 & 38 & 37 & 33 & 67 & 46 & \bf{3} & \bf{2}\\ \hline
University of Virginia & 37 & 38 & 35 & 37 & 33 & 34 & 23 & 28 & 2 & 2\\ \hline
\end{tabular}}
\label{Rankings Table Aggregated}
\end{table}
\FloatBarrier
\begin{table}[!htbp]
\vspace{-11pt}
\caption*{Table \ref{Rankings Table Aggregated} (continued)}
\centering
\Gigantic{54pt}
\resizebox{\columnwidth}{!}{
\begin{tabular}{|p{20cm}|c|c|c|c|c|c|c|c|c|c|}
\hline
\multirow{2}{*}{\textbf{University}} & \multicolumn{8}{c|}{\textbf{Rank}} & \multicolumn{2}{c|}{\textbf{Tier}}\\ \cline{2-11}
& \textbf{BP-FASP} & \textbf{FASP} & \textbf{MVR} & \textbf{MVS} & \textbf{MVS$_2$} & \textbf{HUB} & \textbf{AUTHORITY} & \textbf{U.S.NEWS} & \textbf{D-C} & \textbf{BP-FASP}\\ \hline
State University of New York Buffalo & 38 & 39 & 31 & 29 & 28 & 26 & 33 & 28 & 2 & 2\\ \hline
Iowa State University & 39 & 40 & 47 & 45 & 43 & 41 & 37 & 26 & 2 & 2\\ \hline
Case Western University & \hspace{4pt} 40$^{*}$ & 45 & 53 & 32 & 64 & 36 & 55 & 38 & 2 & 2\\ \hline
\textbf{University of Texas Dallas} & \hspace{4pt} 40$^{*}$ & 42 & 29 & 46 & 60 & 55 & 82 & 58 & \bf{3} & \bf{2}\\ \hline
\textbf{Washington University St. Louis} & \hspace{4pt} 40$^{*}$ & 41 & 43 & 66 & \multicolumn{1}{c|}{67} & 71 & 34 & 39 & \bf{1} & \bf{2}\\ \hline
Arizona State University & 43 & 43 & 49 & 48 & 47 & 21 & 11 & 23 & 2 & 2\\ \hline
University of Washington & 43 & 35 & 40 & 42 & 40 & 42 & 35 & 26 & 2 & 2\\ \hline
\textbf{Clemson University} & \hspace{4pt} 45$^{*}$ & 44 & 58 & 58 & 49 & 61 & 52 & 32 & \bf{3} & \bf{2}\\ \hline
\textbf{Kansas State University} & \hspace{4pt} 45$^{*}$ & 46 & 39 & 50 & 41 & 57 & 48 & 46 & \bf{3} & \bf{2}\\ \hline
\textbf{Texas Tech University} & 47 & 47 & 48 & 47 & 46 & 47 & 68 & 53 & \bf{3} & \bf{2}\\ \hline
\textbf{Missouri University of Science and Technology} & \hspace{4pt} 48$^{*}$ & 50 & 44 & 43 & 51 & 50 & 78 & 58 & \bf{3} & \bf{2}\\ \hline
\textbf{University of Tennessee Knoxville} & \hspace{4pt} 48$^{*}$ & 48 & 69 & 56 & 71 & 68 & 72 & 58 & \bf{3} & \bf{2}\\ \hline
\textbf{University of Massachusetts Amherst} & \hspace{4pt} 50$^{*}$ & 53 & 46 & 39 & 34 & 18 & 9 & 36 & \bf{2} & \bf{3}\\ \hline
Wayne State University & \hspace{4pt} 50$^{*}$ & 51 & 61 & 54 & 54 & 62 & 41 & 53 & 3 & 3\\ \hline
\textbf{George Washington University} & \hspace{4pt} 52$^{*}$ & 54 & 51 & 44 & 45 & 37 & 27 & 53 & \bf{2} & \bf{3}\\ \hline
\textbf{University of Connecticut} & \hspace{4pt} 52$^{*}$ & 55 & 54 & 59 & 44 & 31 & 39 & N/A & \bf{2} & \bf{3}\\ \hline
\textbf{George Mason University} & \hspace{4pt} 54$^{*}$ & 57 & 60 & 51 & 56 & 35 & 42 & 32 & \bf{2} & \bf{3}\\ \hline
Oklahoma State University & \hspace{4pt} 54$^{*}$ & 52 & 38 & 34 & 38 & 32 & 44 & 39 & 3 & 3\\ \hline
\textbf{Northeastern University} & \hspace{4pt} 56$^{*}$ & 56 & 82 & 81 & 53 & 51 & 4 & 36 & \bf{2} & \bf{3}\\ \hline
University of Alabama & \hspace{4pt} 56$^{*}$ & 58 & 45 & 52 & 58 & 63 & 57 & N/A & 3 & 3\\ \hline
Auburn University & \hspace{4pt} 58$^{*}$ & 62 & 52 & 49 & 52 & 52 & 66 & 32 & 3 & 3\\ \hline
University of Louisville & \hspace{4pt} 58$^{*}$ & 59 & 50 & 57 & 57 & 65 & 58 & 66 & 3 & 3\\ \hline
University of Houston & 60 & 60 & 55 & 62 & 61 & 56 & 77 & 53 & 3 & 3\\ \hline
Air Force Academy & 61 & 63 & 77 & 71 & 63 & 54 & 47 & 46 & 3 & 3\\ \hline
\textbf{Naval Postgraduate School} & 62$^{*}$ & 64 & 75 & 83 & 50 & 60 & 14 & 23 & \bf{1} & \bf{3}\\ \hline
University of Oklahoma & 62$^{*}$ & 65 & 56 & 53 & 48 & 58 & 50 & 46 & 3 & 3\\ \hline
\bf{New Jersey Institute of Technology} & \hspace{4pt} 64$^{*}$ & 77 & 80 & 79 & 62 & 80 & 26 & N/A & \bf{2} & \bf{3}\\ \hline
University of Central Florida & \hspace{4pt} 64$^{*}$ & 66 & 57 & 55 & 59 & 59 & 62 & 39 & 3 & 3\\ \hline
Old Dominion University & 66 & 67 & 76 & 67 & 77 & 74 & 83 & N/A & 3 & 3\\ \hline
University of North Carolina Charlotte & 67 & 68 & 73 & 65 & 75 & 73 & 76 & 58 & 3 & 3\\ \hline
West Virginia University & 68 & 69 & 70 & 63 & 69 & 66 & 64 & 66 & 3 & 3\\ \hline
\textbf{Stevens Institute of Technology} & \hspace{4pt} 69$^{*}$ & 74 & 74 & 61 & 55 & 40 & 17 & 39 & \bf{2} & \bf{3}\\ \hline
University of Miami & \hspace{4pt} 69$^{*}$ & 49 & 68 & 75 & 80 & 77 & 79 & 66 & 3 & 3\\ \hline
University of Texas Arlington & \hspace{4pt} 69$^{*}$ & 75 & 66 & 60 & 70 & 46 & 70 & 58 & 3 & 3\\ \hline
Wichita State University & \hspace{4pt} 69$^{*}$ & 70 & 78 & 69 & 68 & 78 & 56 & 66 & 3 & 3\\ \hline
Ohio University & 73 & 71 & 83 & 68 & 78 & 72 & 80 & 66 & 3 & 3\\ \hline
Florida Institute of Technology & \hspace{4pt} 74$^{*}$ & 76 & 62 & 70 & 83 & 79 & 81 & N/A & 3 & 3\\ \hline
Florida State University & \hspace{4pt} 74$^{*}$ & 72 & 79 & 64 & 74 & 64 & 65 & 66 & 3 & 3\\ \hline
Montana State University & \hspace{4pt} 74$^{*}$ & 82 & 59 & 72 & 82 & 76 & 74 & N/A & 3 & 3\\ \hline
New Mexico State University & \hspace{4pt} 74$^{*}$ & 81 & 63 & 73 & 79 & 75 & 73 & N/A & 3 & 3\\ \hline
North Carolina A\&T State University & \hspace{4pt} 74$^{*}$ & 78 & 71 & 78 & 72 & 81 & 81 & 66 & 3 & 3\\ \hline
Oregon State University & \hspace{4pt} 74$^{*}$ & 73 & 72 & 77 & 66 & 69 & 69 & 46 & 3 & 3\\ \hline
State University of New York Binghamton & \hspace{4pt} 74$^{*}$ & 61 & 81 & 80 & 65 & 67 & 53 & 58 & 3 & 3\\ \hline
University of Arkansas Little Rock & \hspace{4pt} 74$^{*}$ & 83 & 67 & 82 & 81 & 70 & 54 & 46 & 3 & 3\\ \hline
University of Wisconsin Milwaukee & \hspace{4pt} 74$^{*}$ & 79 & 65 & 76 & 76 & 83 & 60 & 58 & 3 & 3\\ \hline
\textbf{Worcester Polytechnic Institute} & \hspace{4pt} 74$^{*}$ & 80 & 64 & 74 & 73 & 82 & 46 & 53 & \bf{2} & \bf{3}\\ \hline
\end{tabular}}
\end{table}

To compare our approach with ERGM results, we divide the BP-FASP ranking of IEOR departments into three similarly sized hierarchical clusters (tiers). Due to a tie between the Univ. of Texas Austin and Lehigh Univ. at rank 14, the clustering assigns the top 13 departments to Tier 1, the next 37 to Tier 2, and the remaining 34 to Tier 3. The tiers of the departments obtained with the two approaches---referred henceforth as ``BP-FASP'' and ``D-C'', respectively---are displayed in Table \ref{Rankings Table Aggregated}; departments placed in the different tiers in the two works are shown in bold.
 
Next, we elaborate on the similarities and differences between the two approaches. Carnegie Mellon Univ., Univ. of California Berkeley, Massachusetts Inst. of Technology, Stanford Univ., Cornell Univ., Princeton Univ., Columbia Univ., and Northwestern Univ. are placed in Tier 1 in both works. However, D-C places the Univ. of Texas Austin, Univ. of Southern California, Univ. of Pennsylvania, Univ. of North Carolina Chapel Hill, Washington Univ. St. Louis and Naval Postgraduate School in Tier 1 as well; all these universities are placed in Tier 2 in BP-FASP, except for Naval Postgraduate School, which is placed in Tier 3. On the other hand, BP-FASP places the Univ. of Michigan Ann Arbor, Georgia Inst. of Technology, Pennsylvania State Univ., Purdue Univ. West Lafayette, and Univ. of Illinois Urbana Champaign in Tier 1, which all are placed in Tier 2 in D-C. Furthermore, D-C places Univ. of Massachusetts Amherst, George Washington Univ., Univ. of Connecticut, George Mason Univ., Northeastern Univ., New Jersey Inst. of Technology, Stevens Inst. of Technology, and Worcester Polytechnic Inst. in Tier 2, whereas they all are placed in Tier 3 in BP-FASP. On the other hand, BP-FASP places the Univ. of Iowa, Univ. of Illinois Chicago, Univ. of South Florida, Univ. of Texas Dallas, Clemson Univ., Kansas State Univ., Texas Tech Univ., Missouri Univ. of Science and Technology, and Univ. of Tennessee Knoxville in Tier 2, all of which are placed in Tier 3 in D-C. 

Fig. \ref{D-C and BP-FASP Clusterings of IEOR departments} depicts the clustering/tiering of IEOR departments in D-C and B-FAS. The total number of faculty placements from Tier 1 to Tiers 2 and 3 in BP-FASP (374 placements) are much higher than D-C (196 placements), which is more desirable from a clustering standpoint, as this indicates a higher contrast in the prestige of Tier 1 departments. Surprisingly, there is a significant difference between the number of placements from Tier 1 to Tier 3 in BP-FASP (133 placements) and D-C (14 placements). In contrast, the total number of placements from Tiers 2 and 3 to Tier 1 in BP-FASP (44 placements) is slightly higher than in D-C (38 placements). Despite being slightly smaller in size, the number of placements within departments in Tier 1 in BP-FASP (241 placements) is greater than those in D-C (199 placements), which is more desirable as prestigious departments tend to hire from other prestigious departments. Note that the tier self-loops displayed in Fig. \ref{D-C and BP-FASP Clusterings of IEOR departments} are the total number of faculty hires between the departments in that tier, including individual departments self-hires.
\begin{figure}[th]
    \centering
    \includegraphics[scale=0.4]{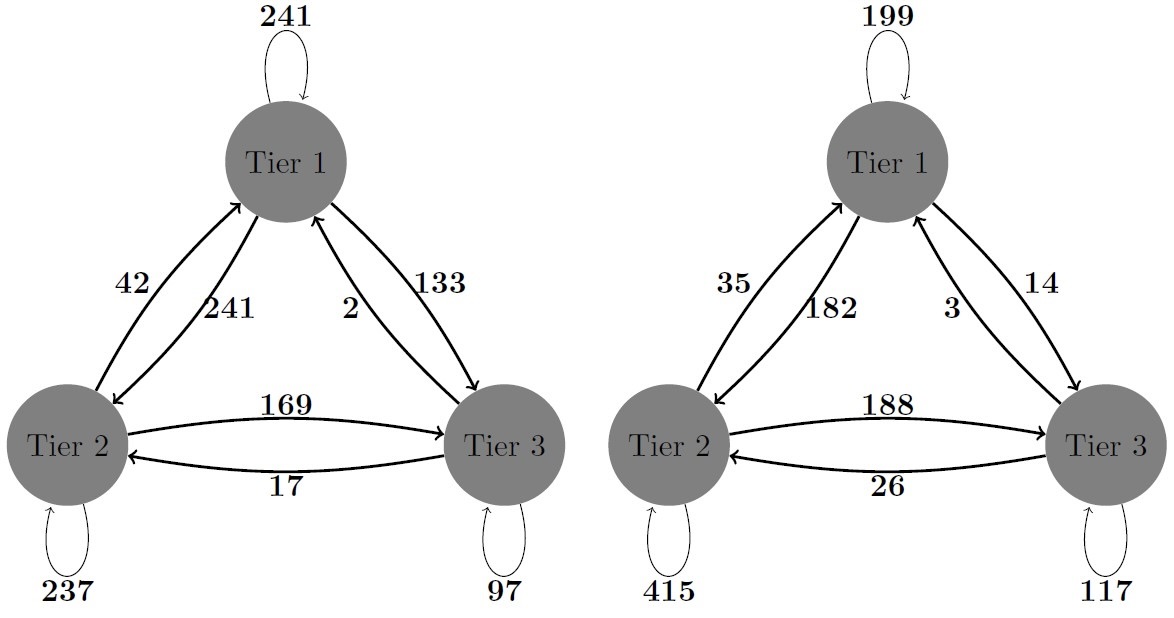}
    \caption{Clustering/tiering of IEOR departments using BP-FASP (left) and D-C (right)}
    \label{D-C and BP-FASP Clusterings of IEOR departments}
\end{figure}

We evaluate these two tierings of IEOR departments based on the distinctiveness of their clusters, which translates into imbalanced hiring and exposes the underlying hierarchical relationships. We use variants of Cut Imbalance (CI), a measure of comparing digraph clustering algorithms \citep{cucuringu2020hermitian}. To introduce them, let $\mathcal{X}$ and $\mathcal{Y}$ be any two disjoint vertex sets. The fundamental CI measure, denoted as $CI^{base}$, is defined as\llV
\begin{small}
\begin{align}
\allowdisplaybreaks
    CI^{base}(\mathcal{X}, \mathcal{Y}) = \frac{1}{2}\frac{\left|w(\mathcal{X}, \mathcal{Y}) - w(\mathcal{Y}, \mathcal{X})\right|}{w(\mathcal{X}, \mathcal{Y}) + w(\mathcal{Y}, \mathcal{X})} \in [0, 1/2],
\end{align}
\end{small}
\noindent where $w(\mathcal{X}, \mathcal{Y}) = \sum_{i \in \mathcal{X}, j \in \mathcal{Y}} w_{ij}$ (the cumulative weights of the arcs going from  $\mathcal{X}$ to $\mathcal{Y}$). In words, $CI^{base}(\mathcal{X}, \mathcal{Y})$ quantifies the imbalance or difference in magnitude between the weight of opposing edges between $\mathcal{X}$ and $\mathcal{Y}$, with $CI^{base}(\mathcal{X}, \mathcal{Y})=1/2$ (resp. $CI^{base}(\mathcal{X}, \mathcal{Y})=0$) indicating that the edges between $\mathcal{X}$ and $\mathcal{Y}$ are completely imbalanced (resp. balanced). Higher cut imbalances are an indicator of more distinctive clusters, which is more desirable. In the extreme case, when the opposing edges have a similar weight, it can be deduced that there are no differences between the clusters.

The size and volume of sets $\mathcal{X}$ and $\mathcal{Y}$ may have an impact on the value of $w(\mathcal{X}, \mathcal{Y})$, potentially making $CI^{base}$ a misleading metric. To mitigate this effect, \citet{cucuringu2020hermitian} proposed two normalized variants, denoted herein as $CI^{size}$ and $CI^{vol}$, which normalize for set size and volume, respectively, to ensure a fair assessment of clustering similarity across sets of varying sizes and volume. \citet{he2022digrac} proposed another variant denoted herein as $CI^{vol\char`_sum}$. The three variants are defined as\llV\lV
\begin{small}
\begin{align}
    &CI^{size}(\mathcal{X}, \mathcal{Y}) = CI(\mathcal{X}, \mathcal{Y}).\min\{\left|\mathcal{X}\right|, \left|\mathcal{Y}\right|\},\\
    &CI^{vol}(\mathcal{X}, \mathcal{Y}) = CI(\mathcal{X}, \mathcal{Y}).\min\{Vol(\mathcal{X}), Vol(Y)\},\\
    &CI^{vol\char`_sum}(\mathcal{X}, \mathcal{Y}) = 2\frac{\left|w(\mathcal{X}, \mathcal{Y}) - w(\mathcal{Y}, \mathcal{X})\right|}{Vol(\mathcal{X}) + Vol(\mathcal{Y})} \in [0,1],
\end{align}
\end{small}
\noindent where $Vol(\mathcal{X})$ is the sum of weighted in-degrees and out-degrees of the vertices in $\mathcal{X}$, excluding self-loop weights. These normalizations help prevent inflated CI values that arise from imbalanced comparisons. In particular, when one of the sets is very small (or has low volume), achieving a high $CI^{base}$ becomes trivially easy, even if the clustering quality is poor. For instance, it is much easier to achieve the maximum $CI^{base}$ when $|\mathcal{X}| = 100$ and $|\mathcal{Y}| = 1$, compared to when $|\mathcal{Y}| = 50$. Similar to $CI^{base}$, the higher the values of these measures, the more distinct the clusters are. 

Table \ref{CI Table} displays different CI measures for BP-FASP and D-C for each pair of tiers and for the overall clustering; the best pairwise and overall CI-measures are shown in bold. BP-FASP not only has a higher overall CI-measure according to each of the four CI-measures, it has a higher value for each pair of tiers, except for CI$^{size}$(Tier 1, Tier 2). The difference in overall CI-values is notable, especially for CI$^{vol}$ and $CI^{vol\_sum}$, where BP-FASP values are 31\% and 41\% higher, respectively.\lV
\begin{table}[thbp]
\centering
\footnotesize
\caption{Various CI measures for the D-C and B-FASP clustering. The highest value for each CI measure and each pair of tiers is shown in bold}
\resizebox{\columnwidth}{!}{
\begin{tabular}{|l|ll|ll|ll|ll|}
\hline
 & \multicolumn{2}{l|}{$CI^{base}$}    & \multicolumn{2}{l|}{$CI^{size}$}  & \multicolumn{2}{l|}{$CI^{vol}$}    & \multicolumn{2}{l|}{$CI^{vol\_sum}$}    \\ \hline
 & \multicolumn{1}{l|}{B-FAS} & \multicolumn{1}{l|}{D-C} & \multicolumn{1}{l|}{B-FAS} & \multicolumn{1}{l|}{D-C} & \multicolumn{1}{l|}{B-FAS} & \multicolumn{1}{l|}{D-C} & \multicolumn{1}{l|}{B-FAS} &  \multicolumn{1}{l|}{D-C}\\ \hline
(Tier 1, Tier 2) & \multicolumn{1}{l|}{\textbf{0.352}} & \multicolumn{1}{l|}{0.339} & \multicolumn{1}{l|}{4.571} & \multicolumn{1}{l|}{\textbf{4.742}} & \multicolumn{1}{l|}{\textbf{281.272}} & \multicolumn{1}{l|}{182.226} & \multicolumn{1}{l|}{\textbf{0.247}} &  \multicolumn{1}{l|}{0.178}\\ \hline
(Tier 1, Tier 3) & \multicolumn{1}{l|}{\textbf{0.485}} & \multicolumn{1}{l|}{0.324} & \multicolumn{1}{l|}{\textbf{6.307}} & \multicolumn{1}{l|}{4.529} & \multicolumn{1}{l|}{\textbf{201.352}} & \multicolumn{1}{l|}{122.618} &  \multicolumn{1}{l|}{\textbf{0.216}} &  \multicolumn{1}{l|}{0.024}\\ \hline
(Tier 2, Tier 3) & \multicolumn{1}{l|}{\textbf{0.409}} & \multicolumn{1}{l|}{0.379} & \multicolumn{1}{l|}{\textbf{13.892}} & \multicolumn{1}{l|}{12.869} & \multicolumn{1}{l|}{\textbf{169.570}} & \multicolumn{1}{l|}{143.453} & \multicolumn{1}{l|}{\textbf{0.248}} &  \multicolumn{1}{l|}{0.217}\\ \hline
Total & \multicolumn{1}{l|}{\textbf{1.246}} & \multicolumn{1}{l|}{1.042} & \multicolumn{1}{l|}{\textbf{24.770}} & \multicolumn{1}{l|}{22.14} & \multicolumn{1}{l|}{\textbf{652.194}} & \multicolumn{1}{l|}{448.297} & \multicolumn{1}{l|}{\textbf{0.711}} &  \multicolumn{1}{l|}{0.419}\\ \hline
\end{tabular}}
\label{CI Table}
\end{table}

We close this section by discussing self-hires, whose exclusion poses a limitation to MVR, BP-FASP, and their variants in analyzing FHN. Arguably, as the prestige of departments increases, the pool of candidates from higher-ranked departments decreases. Logically, when a department hires one of its Ph.D. graduates, this may be considered a sign of high prestige since low-prestige departments have more candidates to choose from compared to high-prestige departments. This argument can be backed by the data in the IEOR case, where the top self-hiring departments are Massachusetts Inst. of Technology and Stanford Univ. with 20 and 10 self-hires, respectively. There are 165 self-hires in the IEOR FHN, out of which 50 self-hires happen within Tier 1, 65 in Tier 2, and 50 in Tier 3 in BP-FASP. On the other hand, 47 self-hires happen within Tier 1, 75 in Tier 2, and 43 in Tier 3 in D-C. As such, BP-FASP's Tier 1 accounts for more self-hires, despite its slightly smaller size; on the other hand, BP-FASP's Tier 3 includes more self-hires. The higher-than-expected number of self-hires in Tier 3 in both works likely stems from a broad definition of IEOR, where only 47.5\% are immediate and/or faculty trained at the same university come from departments of varying prestige.\vspace{-0.2in} 
\vspace{-3pt}
\section{Conclusion}\label{Sec: Conclusion}\vspace{-0.1in}
This paper modifies the traditional feedback arc set problem (FASP) to allow for certain types of cycles to be preserved from digraphs containing bidirected arcs. The traditional feedback set problem aims to find the set of arcs with the least cumulative weight such that removing them results in a directed acyclic graph (DAG). On the other hand, the modified feedback set problem aims to find the set of arcs with the least cumulative weight such that removing them results in a unicycle-free subgraph. This modification allows the associated ranking to contain ties and, more broadly, equivalence classes. To do so, the paper extends the relationship between DAGs and linear orderings to the case of unicycle-free graphs and weak orderings. Moreover, it presents two mathematical formulations and compares their strength via polyhedral analysis. Furthermore, it analyzes the IEOR faculty hiring network. Apart from obtaining the rank of IEOR departments, we cluster the IEOR departments into three tiers, similar to an existing clustering approach based on ERGM. Our results show that our approach results in more distinctive tiers according to the cut imbalance criterion and its variants.

Future work will devise new methodologies to account for self-hires. Moreover, it will consider using different weights for faculty placements at different points in time,  giving recent placements more weight. These adjustments could be useful since faculty placements happen over a wide range of time and also tend to be influenced by the prestige of departments at the time of occurrence. Hence, the observed network is affected by the prestige of the departments over time, and not just in recent years. Approaches for aggregating the rankings obtained with different methodologies (e.g., \citep{cook2006distance} will also be explored. 

The data that support the findings of this study are available from the corresponding author upon reasonable request.\llV\lV
%Although, FHN provides an alternative methodology for determining the prestige/rank of academic departments, multiple rankings can be aggregated into one \emph{consensus ranking} to use the information from all sources. This can be accomplished by ranking aggregations methods such as Kemeny aggregation (e.g., \cite{akbari2023beyond}).

%In an ideal scenario, the faculty hiring network must prioritize the immediate movements of faculty members after their graduation, rather than their current positions. However, most works, including \citet{del2020exponential}, relax this assumption due to data limitations. They collected data based on the current positions of faculty members, neglecting their initial post-graduation positions. Another challenge associated with this approach is that some faculty members may have commenced their careers in lower-ranked departments than their alma mater, but over time, they managed to ascend to higher-ranked departments than their academic alma mater. Apart from the initial post-graduation positions, subsequent movements of the faculties can be included as well. Other educational outputs such as research impact, number/quality of citations, etc., can also be considered. Future studies will update the IEOR dataset or replicate the study for FHNs of other fields.  

% \small
%\Gigantic{9.57pt}\setlength{\bibsep}{2pt}
\bibliographystyle{plainnat}
\bibliography{refs}

\newpage
\appendix
\section*{Appendix}
\section{Depiction of Trivially Unicycle-Free Graphs}\label{Appendix Trivial unicycle-free graphs of size 3}
Each of the shown graphs cannot contain a cycle since it has at most two arcs.
\begin{figure}[H]
\centering
\begin{subfigure}{0.225\textwidth}
\begin{tikzpicture}
\tikzstyle{every node}=[draw, shape=circle];
\path (0, 0cm) node (j) {$j$};
\path (2.2, 0cm) node (k) {$k$};
\path (1.1, 1.6cm) node (i) {$i$};
\draw[line width=1pt,->] 
(j) to (i);
\draw[line width=1pt,->]
(k) to (i);
\end{tikzpicture}
\caption{}
\end{subfigure}
\begin{subfigure}{0.225\textwidth}
\begin{tikzpicture}
\tikzstyle{every node}=[draw, shape=circle];
\path (0, 0cm) node (j) {$j$};
\path (2.2, 0cm) node (k) {$k$};
\path (1.1, 1.6cm) node (i) {$i$};
\draw[line width=1pt,->] 
(i) to (j);
\draw[line width=1pt,->]
(i) to (k);
\end{tikzpicture}
\caption{}
\end{subfigure}
\begin{subfigure}{0.225\textwidth}
\begin{tikzpicture}
\tikzstyle{every node}=[draw, shape=circle];
\path (4, 0cm) node (j) {$j$};
\path (6.2, 0cm) node (k) {$k$};
\path (5.1, 1.6cm) node (i) {$i$};
\draw[line width=1pt,->]
(i) to (j);
\end{tikzpicture}
\caption{}
\end{subfigure}
\begin{subfigure}{0.225\textwidth}
\begin{tikzpicture}
\tikzstyle{every node}=[draw, shape=circle];
\path (4, 0cm) node (j) {$j$};
\path (6.2, 0cm) node (k) {$k$};
\path (5.1, 1.6cm) node (i) {$i$};
\end{tikzpicture}
\caption{}
\end{subfigure}
\caption{Trivial unicycle-free graphs of size three}
\label{Trivial unicycle-free graphs of size 3}
\end{figure}
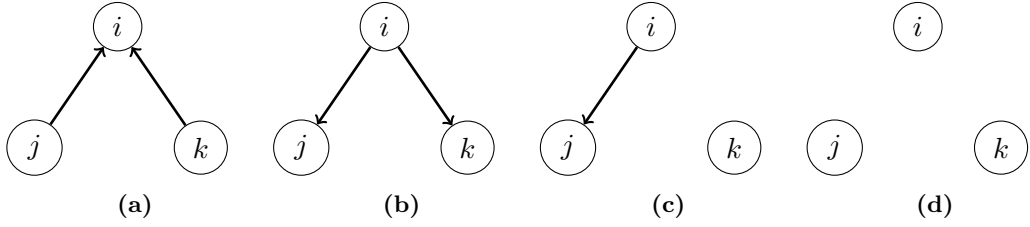

%\appendix
\section{Proof of Theorem 1}\label{Appendix Proof of Slater's Theorem}
\begin{manualtheorem}{1}
Let $\cW$ define a weak tournament over a set of candidates $\cV$. Define a weak tournament solution $S$ as the alternative(s) in the maximal subset of the $K$-way partition obtained by making the minimum number of majority relation reversals to $\cW$. Then, $S$ is an extension of Slater's rule to weak tournaments.
\end{manualtheorem}
\begin{proof}
When $\cW\in\cT$, either $a\succ b$ or $b\succ a$ (but not both), for all $a,b\in\cV$. Consequently, there are no bidirectional subgraphs in $\cT$, and any $K$-way partition $\boldsymbol{\mathcal{N}}$ will have $|\mathcal{N}_{t}| = 1, \hspace{5pt} \forall \hspace{1pt}\mathcal{N}_{t} \in \boldsymbol{\mathcal{N}}$. Therefore, finding the minimum number of majority relation reversals needed to obtain such a partition and returning its maximal (singleton) subset is equivalent to applying the standard Slater's rule to $\cT$. 
\end{proof}

%\appendix
\section{Proof of Proposition 1} \label{Appendix Proof of Proposition 1}
\begin{manualProposition}{1}
BLP \#1 admits any weak ordering that can be associated with a unicycle-free subgraph within a given digraph $\mathcal{G} = (\mathcal{N}, \mathcal{A})$.
\end{manualProposition}
\begin{proof}
The sets  $\{ y_{ij} \mid y_{ij} + y_{ji} = 1, \quad y_{ij} - y_{ik} - y_{kj} \geq -1, \quad y_{ij} \in \{0, 1\}, \quad \forall i \neq j \neq k \in \mathcal{N}\}$ and $\{ y_{ij} \mid y_{ij} + y_{ji} \geq 1, \quad y_{ij} - y_{ik} - y_{kj} \geq -1, \quad y_{ij} \in \{0, 1\}, \quad \forall i \neq j \neq k \in \mathcal{N}\}$ admit all linear and weak orderings of a given set $\mathcal{N}$, respectively %\citep{yoo2021new} %% Replace after paper is accepted
(note that the second set is a super set of the first). The constraints of BLP \#1 admit any weak ordering that can be associated with any unicycle-free subgraph of $\mathcal{G}$ by simply excluding those that cannot associated; that is, by imposing a strict ordering on the those pairs of nodes for which no bidirectional arc exists in $\mathcal{G}$.
\end{proof}

\section{Proof of Proposition 2} \label{Appendix Proof of transitivity}
\begin{manualProposition}{2}
The transitivity requirement can be imposed by Constraints (2).
\end{manualProposition}
\begin{proof}
Consider an arbitrary item-triplet $(i, j, k)$. By examining all three unicycles in Fig. 3, we observe a common pattern: there are arcs $(i, j)$, $(j, k)$, and $(k, i)$, but not no arc $(j, i)$. In terms of our variables, there is a unicycle in the solution if $y_{ij}=y_{jk}=y_{ki} = 1, y_{ji}=0$. We show that Constraints (2) are violated with such a setting. Note that we do not assume any values for the $y_{ik}$ and $y_{kj}$ variables, as their exact values determine different type of unicycles shown in Fig. 4. More specifically, $y_{ik} = y_{kj} = 0$, $y_{ik} = y_{kj} = 1$, and $y_{ik} = 0, y_{kj} = 1$ induce different cycles.

To evaluate the transitivity between the three items, consider Constraint (2a) for the item-pair $(k, i)$, Constraint (2b) for the item-pair $(j, k)$, and Constraint (2c) for the item-pair $(j, i)$. We have
\allowdisplaybreaks
\begin{subequations}\label{eq: 9261}
\begin{align}
    &\sum_{u \in \mathcal{N}\backslash \{i, k\}}(y_{ui} - y_{uk}) \geq (n - 2)(y_{ki} - 1),  \label{eq: 261}\\
    &\sum_{u \in \mathcal{N}\backslash \{j, k\}}(y_{ju} - y_{ku}) \geq (n - 2)(y_{jk} - 1), \label{eq: 262}\\
    &\sum_{u \in \mathcal{N}\backslash \{i, j\}}(y_{ju} + y_{ui}) \leq (n - 2)(y_{ji} + 1).\label{eq: 263}
\end{align}
\end{subequations}
Constraints. \eqref{eq: 261}-\eqref{eq: 263} can be rewritten as:
\allowdisplaybreaks
\begin{subequations}\label{eq: 4361}
\begin{align}
( y_{ji} - y_{jk}) + \left(\sum_{u \in \mathcal{N}\backslash \{i, j, k\}}y_{ui} - y_{uk}\right) \geq (n - 2)(y_{ki} - 1),\label{eq: 361}\\
(y_{ji} - y_{ki}) + \left(\sum_{u \in \mathcal{N}\backslash \{i, j, k\}}y_{ju} - y_{ku}\right) \geq (n - 2)(y_{jk} - 1),\label{eq: 362}\\
(y_{jk} + y_{ki}) + \left(\sum_{u \in \mathcal{N}\backslash \{i, j, k\}}y_{ju} + y_{ui}\right) \leq (n - 2)(y_{ji} + 1)\label{eq: 363}.
\end{align}
\end{subequations}

The current settings of variables $y_{ij}, y_{jk}, y_{ki},$ and $y_{ji}$ without assuming any specific values for the remaining variables, yields
\allowdisplaybreaks
\begin{small}
\begin{subequations}
\begin{align*}
    (0 - 1) + \left(\sum_{u \in \mathcal{N}\backslash \{i, j, k\}}y_{ui} - y_{uk}\right) \geq (n - 2)(1 - 1),\\
    (0 - 1) + \left(\sum_{u \in \mathcal{N}\backslash \{i, j, k\}}y_{ju} - y_{ku}\right) \geq (n - 2)(1 - 1),\\
    (1 + 1) + \left(\sum_{u \in \mathcal{N}\backslash \{i, j, k\}}y_{ju} + y_{ui}\right) \leq (n - 2)(0 + 1).
\end{align*}
\end{subequations}
\end{small}
Simplifying the above expressions yields the linear inequalities 
\allowdisplaybreaks
\begin{small}
\begin{subequations}\label{equation 643}
\begin{align}
    &\sum_{u \in \mathcal{N}\backslash \{i, j, k\}}y_{ui} \geq 1 + \sum_{u \in V\backslash \{i, j, k\}}y_{uk},\label{eq: 23}\\
    &\sum_{u \in \mathcal{N}\backslash \{i, j, k\}}y_{ju} \geq 1 + \sum_{u \in V\backslash \{i, j, k\}}y_{ku},\label{eq: 24}\\
    &\sum_{u \in \mathcal{N}\backslash \{i, j, k\}}(y_{ju} + y_{ui}) \leq n - 4 \label{eq: 25}.
\end{align}
\end{subequations}
\end{small}
We show that this set of linear inequalities cannot be satisfied by adding \eqref{eq: 23} and \eqref{eq: 24}, which gives
\allowdisplaybreaks
\begin{small}
\begin{subequations}
\begin{align}
    \sum_{u \in \mathcal{N}\backslash \{i, j, k\}}(y_{ju} + y_{ui}) &\geq  2 + \sum_{u \in \mathcal{N}\backslash \{i, j, k\}}(\underbrace{y_{uk} + y_{ku}}_\text{$\geq 1$}) \geq   2 + (n-3) = n-1,\label{eq: 19}
\end{align}
\end{subequations}
\end{small}
\noindent where inequality \eqref{eq: 19} leverages Constraints (1b)-(1c).
Inequality \eqref{eq: 19} contradicts inequality \eqref{eq: 25}, hence, the solution is infeasible. Thereby, we can conclude that Constraints (2) prevent unicycles and enforce the transitivity relation in BLP \#1. 
\end{proof}

\section{Proof of Theorem 2} \label{Appendix Proof of Polyhedral Theorem}
\begin{manualtheorem}{1}
Let $\mathcal{P}^1$ and $\mathcal{P}^2$ denote the LP-relaxed regions of BLPs \#1 and \#2, respectively. For any instance of BP-FASP, $\mathcal{P}^1 \subseteq \mathcal{P}^2$, and this inclusion can be strict.
\end{manualtheorem}
\begin{proof}
We prove that $\mathcal{P}^1 \subseteq \mathcal{P}^2$ by showing that every feasible solution to BLP \#1 is also feasible to BLP \#2. Note that BLPs \#1 and \#2 only differ in their preference transitivity constraints; hence, we only focus on comparing these constraints. Let $\boldsymbol{y}^{(1)} \in \mathcal{P}^{1}$ and $\boldsymbol{y}^{(2)} \in \mathcal{P}^{2}$ denote solutions that are feasible for BLP \#1 and BLP \#2, respectively.

To show that $\boldsymbol{y}^{(1)} \in \mathcal{P}^{2}$, we re-organize the terms of Constraints (2a)-(2c) to obtain the left-hand side expression of Constraint (1d). Afterward, we compare the right-hand side values. First, consider Constraint (2a) for item-pair $(k, j)$. This constraint is re-organized as follows:
\allowdisplaybreaks
\begin{small}
\begin{align*}
\Rightarrow \sum_{u \in V\backslash \{j, k\}}(y_{uj}^{(1)} - y_{uk}^{(1)}) & \geq (n - 2)(y_{kj}^{(1)} - 1) \\
\Rightarrow (y_{ij}^{(1)} - y_{ik}^{(1)}) + \sum_{u \in V\backslash \{i, j, k\}}(y_{uj}^{(1)} - y_{uk}^{(1)}) & \geq (n - 2)(y_{kj}^{(1)} - 1)\\
\Rightarrow y_{ij}^{(1)} - y_{ik}^{(1)}  - y_{kj}^{(1)} & \geq (n - 3)y_{kj}^{(1)} - (n - 2) + \sum_{u \in V\backslash \{i, j, k\}}(y_{uk}^{(1)} - y_{uj}^{(1)}),\\&
= \underbrace{- (n - 2)  + \underbrace{\sum_{u \in V\backslash \{i, j, k\}} \underbrace{(y_{kj}^{(1)} + y_{uk}^{(1)} - y_{uj}^{(1)}}_\text{$\leq 1$})}_\text{$\leq n-3$}}_\text{$\leq -1$}.
\end{align*}
\end{small}
The last inequality leverages the fact $\boldsymbol{y}^{(1)}$ satisfies Constraint (1d). Since the maximum value of the right-hand side of the last equation is less than or equal to -1, we can conclude that every feasible solution for Constraint (1d) is also feasible for Constraint (2a). The proofs for Constraints (2b)-(2c) are similar and are thereby omitted for succinctness. Therefore, we can conclude that $\boldsymbol{y}^{(1)} \in \mathcal{P}^{2}$, and hence, $\mathcal{P}^1 \subseteq \mathcal{P}^2$. 

To show that the inclusion $\mathcal{P}^{1} \subseteq \mathcal{P}^{2}$ can be strict (i.e., $\mathcal{P}^{1} \subset \mathcal{P}^{2}$), we find a point in $P_{2}\backslash P_{1}$ (a point that satisfies Constraint (2a), but not Constraint (1d)). Consider the small graph shown in Fig. \ref{An example graph where }
\begin{figure}[ht]
    \centering
    \includegraphics[scale=0.6]{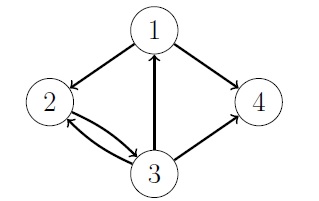}
    \caption{An example graph for which $\mathcal{P}^1 \subset \mathcal{P}^2$}
    \label{An example graph where }
\end{figure}
For this instance, consider the solution $\boldsymbol{y}^{(2)} \in \mathcal{P}^{2}$ as 
\begin{small}
\begin{align*}
\allowdisplaybreaks
&y_{14}^{(2)} = y_{24}^{(2)} = y_{34}^{(2)} = 1, \hspace{15pt}
y_{41}^{(2)} = y_{42}^{(2)} = y_{43}^{(2)} = 0, \hspace{15pt} y_{12}^{(2)} = 0.1,\\& 
y_{21}^{(2)} = 0.9, \hspace{10pt}y_{23}^{(2)} = 0.4, \hspace{10pt} y_{32}^{(2)} = 0.6, \hspace{10pt} y_{13}^{(2)} = 0.6, \hspace{10pt} y_{31}^{(2)} = 0.4.
\end{align*}
\end{small}
Note that these values are fractional as they form a solution to the LP-relaxation of BLP \#2. By inspection, this solution satisfies all constraints of BLP \#2. However, we have that
\begin{small}
\begin{align*}
    y_{12}^{(2)} - y_{32}^{(2)} - y_{13}^{(2)} = 0.1 - 0.6 -0.6 = -1.1 \ngeq -1,
\end{align*}
\end{small}
which indicates that this solution does not satisfy the preference transitivity constraints of BLP \#1. 
\end{proof} 
\end{document}